\documentclass[a4paper,11pt,fleqn]{article} 

\usepackage{amsmath,amssymb,amsthm} 
\usepackage{hyperref}

\allowdisplaybreaks

\newcommand{\p}{\partial}
\newcommand{\e}{\varepsilon}
\newcommand{\const}{\mathop{\rm const}\nolimits}
\newcommand{\sign}{\mathop{\rm sgn}\nolimits}
\renewcommand{\d}{{\,\rm d}}
\newcommand{\lsemioplus}{\mathbin{\mbox{$\lefteqn{\hspace{.70ex}\rule{.4pt}{1.2ex}}{\in}$}}}

\newcommand{\setSkips}{\vspace*{-1mm plus 0mm minus 0.2mm}\parskip=0mm\parsep=0mm\itemsep=0.75mm\topsep=-0.1mm}

\newcommand{\todo}[1][\null]{\ensuremath{\clubsuit}}
\newcommand{\noprint}[1]{}

\newtheorem{theorem}{Theorem}
\newtheorem{lemma}[theorem]{Lemma}
\newtheorem{corollary}[theorem]{Corollary}
\newtheorem{proposition}[theorem]{Proposition}
{\theoremstyle{definition}
\newtheorem{definition}[theorem]{Definition}

\newtheorem{remark}[theorem]{Remark}
}

\begin{document}

\par\noindent {\LARGE\bf Extended symmetry analysis\\ of~generalized Burgers equations \par}

{\vspace{4mm}\par\noindent {\large Oleksandr A. Pocheketa$^\dag$ and Roman O. Popovych$^\ddag$} \par\vspace{2mm}\par}

\vspace{2mm}\par\noindent
$^\dag${\it Center for evaluation of activity of research institutions and scientific support of regional\\ 
$\phantom{^\dag}$development of Ukraine, NAS of Ukraine, 54 Volodymyrska Str., 01030 Kyiv, Ukraine}

\vspace{2mm}\par\noindent 
$^\ddag${\it Wolfgang Pauli Institute, Oskar-Morgenstern-Platz 1, A-1090 Vienna, Austria}\\
$\phantom{^\dag}${\it Institute of Mathematics of NAS of Ukraine, 3~Tereshchenkivska Str., 01004 Kyiv, Ukraine}

\vspace{2mm}\noindent 
$\phantom{^\dag}$E-mail: $^\dag$pocheketa@nas.gov.ua, $^\ddag$rop@imath.kiev.ua

\vspace{7mm}\par\noindent\hspace*{8mm}\parbox{144mm}{\small
Using enhanced classification techniques,
we carry out the extended symmetry analysis of the class of generalized Burgers equations
of the form $u_t+uu_x+f(t,x)u_{xx}=0$.
This enhances all the previous results on symmetries of these equations
and includes the description of admissible transformations, Lie symmetries,
Lie and nonclassical reductions, hidden symmetries, conservation laws,
potential admissible transformations and potential symmetries.
The study is based on the fact that the class is normalized,
and its equivalence group is finite-dimensional.
}\par\vspace{2mm}


\section{Introduction}

The methods of group analysis have been comprehensively developed
and applied to wide and complicated classes of differential equations.
Nevertheless, there are some simple famous model equations
and classes of differential equations that have not been properly investigated
from the symmetry point of view yet.

In this paper we exhaustively classify Lie symmetries and Lie reductions,
reduction operators, conservation laws and potential symmetries
of equations from the delightedly simple class
\begin{gather}\label{GBE}
u_t+uu_x+f(t,x)u_{xx}=0
\qquad
\text{with}
\qquad
f\ne0,
\end{gather}
which are called generalized Burgers equations.
This enhances and essentially extends known results on these equations.
We combine modern methods of group analysis of differential equations
with original techniques,
which include 
the algebraic method for group classification of normalized classes of differential equations~\cite{bihl2012b,popo2010a},
the special technique of classifying appropriate subalgebras~\cite{bihl2012b,kuru2016a},
the classification of Lie reductions for normalized classes of differential equations
up to equivalence transformations,
the selection of optimal ansatzes for reduction~\cite{fush1994a,fush1994b,popo1995c},
mappings between classes generated by families of point transformations~\cite{vane2009a}
and the classification of reduction operators up to admissible transformations~\cite{popo2007a}.

In~\eqref{GBE} and in what follows subscripts of functions denote derivatives with respect to the corresponding variables.
The consideration is within the local framework.
The equation of the form~\eqref{GBE} with a~fixed~$f$, 
which is assumed to be nonvanishing for all $(t,x)$'s from the related domain,
is denoted by~$\mathcal L_f$.

The classical Burgers equation~$\mathcal L_{-\mu}$ with positive constant~$\mu$
was suggested in the 1930s as a~model for one-dimensional turbulence~\cite{burg1948a}.
Among its most famous solutions there are so-called traveling waves.
In~\cite{ligh1956a} it is shown that, under proper interpretation, the same equation describes the propagation
of one-dimensional weak planar waves.
A~formal multiple-scales method was employed to extend this model to weak cylindrical and spherical
waves in~\cite{leib1974a}.
Thus, generalized Burgers equations are used to model a~wide variety of phenomena in physics,
chemistry, mathematical biology, etc.; see, e.g.,~\cite[Chapter~4]{whit1974a}.

As important but simple model, the Burgers equation was intensively studied
and was used as illustrative and toy example or standard test benchmark
in the course of developing various mathematical concepts and methods,
in particular, in the field of group analysis of differential equations;
see \cite{ames1972a,olve1993b} and references therein.
This tradition spread to generalized Burgers equations.
The study of admissible transformations for the class~\eqref{GBE}
in~\cite{king1991c} was a~pioneer work on such transformations in the literature.
It appeared that the equivalence groupoid constituted by admissible transformations
can be described in terms of normalization.
More specifically,
the class~\eqref{GBE} is normalized with respect to the usual equivalence group,
which is six-dimensional,
and this fact is of principal value for the entire consideration in this paper.

Equations from the class~\eqref{GBE} and similar classes were subjects of many papers.
In particular, Lie symmetries and similarity solutions of equations of the form~\eqref{GBE} with $f_x=0$
were considered in~\cite{doyl1990a,wafo2004d}.
It was also shown in~\cite{wafo2004d} that among these equations only
the classical Burgers equation, for which $f=\const$, admits
nontrivial potential symmetries
and regular reduction operators inequivalent to Lie symmetries.
More references and detail comments
on studies related to group analysis of generalized Burgers equations
are given in the corresponding sections of this paper.
Since an essential part of such results presented in the literature
are not exhaustive or completely correct,
it is in fact necessary to carry out extended symmetry analysis of equations from the class~\eqref{GBE}
from the very beginning. 

The structure of the paper is as follows.
Basic notions of group analysis of differential equations
are briefly reviewed in Section~\ref{Theorsection}.
The equivalence groupoid, the equivalence group and the equivalence algebra
of the class~\eqref{GBE} are computed in Section~\ref{GROUPandALGEBRAsection}.
In Section~\ref{LieSYMsection} we exhaustively classify Lie symmetries of equations from the class~\eqref{GBE}
using the algebraic method of group classification.
This method is especially effective for the class~\eqref{GBE}
due to this class is normalized and its equivalence group is finite-dimensional.
Section~\ref{SOLUTIONSsection} is devoted to Lie reductions, hidden symmetries and similarity solutions
of generalized Burgers equations that admit nonzero Lie invariance algebras.
Section~\ref{ROsection} deals with reduction operators and nonclassical reductions
of equations from the class~\eqref{GBE}.
Conservation laws and potential symmetries
of these equations as well as potential admissible transformations between them
are studied in Section~\ref{CLPSsection}.
Implications of paper's results are discussed in the last section.

\section{Theoretical background}\label{Theorsection}

Recall the concepts of a~class of differential equations, the equivalence groupoid
and the equivalence group of a~class, and others involved in the present study.
As in this paper we consider the class of equations of the simple form~\eqref{GBE},
we give definitions for the specific case
of a~class of single second-order partial differential equations
with the two independent variables~$(t,x)$ and the single dependent variable~$u$
that are merely parameterized by the single arbitrary element $f=f(t,x,u)$
without any derivatives of~$f$.
For the general definition of a class of systems of differential equations
involving a~tuple of arbitrary elements with their derivatives see, e.g.,
\cite{bihl2012b,popo2010a}.
We deal here with usual equivalence groups, unless another type of an equivalence group is stated.

Consider a~family of differential equations~$\mathcal L_f$:
$L^f[u]:=L(t,x,u,u_t,u_x,u_{tt},u_{tx},u_{xx},f)=0$
parameterized by a~parameter-function~$f=f(t,x,u)$
running through the set~$\mathcal S$
of solutions of an auxiliary system of differential equations and differential inequalities on~$f$,
where all the variables~$t$, $x$ and~$u$ are assumed to be independent.

\begin{definition}
The set $\{\mathcal L_f\mid f\in\mathcal S\}$ denoted by~$\mathcal L|_{\mathcal S}$ is called a~\emph{class of differential equations},
that is defined by the parameterized form of equations $L(t,x,u,u_t,u_x,u_{tt},u_{tx},u_{xx},f)=0$
and the set~$\mathcal S$ of values of the arbitrary element~$f$.
\end{definition}

For the class~\eqref{GBE}, the form of equations and the set of arbitrary elements are
\begin{gather*}
L^f[u]:=u_t+uu_x+fu_{xx}=0
\qquad
\text{and}
\qquad
\mathcal S=\big\{f=f(t,x,u)\mid f_u=0, f\ne0\big\},
\end{gather*}
respectively.
In general, the arbitrary elements may also depend on derivatives of dependent variables, but for the class~\eqref{GBE}
we have
$f_{u_t}=f_{u_x}=f_{u_{tt}}=f_{u_{tx}}=f_{u_{xx}}=0$.
These constraints will be used
implicitly, e.g.\ while solving determining equations.

A point transformation in the space of $(t,x,u)$
has the form
\begin{gather}
\label{TXU}
\varphi\colon
\quad
\tilde t=T(t,x,u),
\quad
\tilde x=X(t,x,u),
\quad
\tilde u=U(t,x,u),
\end{gather}
where $T$, $X$ and~$U$ are smooth functions of~$t$, $x$ and~$u$
with $|\partial(T,X,U)/\partial(t,x,u)|\ne0$.
Given two fixed equations $\mathcal L_f$ and~$\mathcal L_{\tilde f}$
from the class~$\mathcal L|_{\mathcal S}$ with arbitrary elements~$f$ and $\tilde f$,
by $T(f,\tilde f)$ we denote the set of point transformations in the space $(t,x,u)$
that map $\mathcal L_f$ to~$\mathcal L_{\tilde f}$.
An {\it admissible transformation}~\cite{popo2006b,popo2010a} in the class~$\mathcal L|_{\mathcal S}$
is a~triple consisting of two arbitrary elements \smash{$f,\tilde f\in \mathcal S$}
(or, in other words, the corresponding two equations, which
are called the initial one and the target one)
and a~point transformation $\varphi \in \mathrm T(f,\tilde f)$.

The notion of admissible transformation is a~formalization
of the earlier notions of form-preserving~\cite{king1991c,king1998a}, or allowed~\cite{gaze1992a} 
transformations.

\begin{definition}
The \textit{equivalence groupoid}~$\mathcal G^\sim=\mathcal G^\sim(\mathcal L|_{\mathcal S})$
of the class~$\mathcal L|_{\mathcal S}$ is the set of admissible transformations of this class,
$\{(f,\varphi,\tilde f)\mid f,\tilde f\in\mathcal S,\, \varphi\in\mathrm T(f,\tilde f)\}$,
equipped with the operation ``$\circ$'' of composition of admissible transformations.
\end{definition}

The composition ``$\circ$'' of admissible transformations
$(f_1,\varphi_1,\tilde f_1)$ and $(f_2,\varphi_2,\tilde f_2)$
is defined only if $\tilde f_1=f_2$,
and its result is $(f_1,\varphi_2\varphi_1,\tilde f_2)$.
It is obvious that the axioms of groupoid hold for~$\mathcal G^\sim$:
\begin{enumerate}\setSkips
 \item
$((f_1,\varphi_1,f_2)\circ(f_2,\varphi_2,f_3))\circ(f_3,\varphi_3,f_4)
=(f_1,\varphi_1,f_2)\circ((f_2,\varphi_2,f_3)\circ(f_3,\varphi_3,f_4))$,
which means the associativity of the composition.
 \item
For each~$f$ the role of the neutral element is played by the triple $(f,{\rm id},f)$,
where ${\rm id}$ is the identical transformation, $\tilde t=t$, $\tilde x=x$, $\tilde u=u$.
 \item
Any admissible transformation $(f,\varphi,\tilde f)$ is invertible,
and the inverse is $(\tilde f,\varphi^{-1},f)$.
\end{enumerate}

\begin{definition}
The \emph{usual equivalence group} $G^\sim=G^\sim(\mathcal L|_{\mathcal S})$ of the class~$\mathcal L|_{\mathcal S}$
is the (pseudo)group of point transformations in the extended space of $(t,x,u,f)$,
\begin{gather*}
\mathcal T\colon
\quad
\tilde t=T(t,x,u),
\quad
\tilde x=X(t,x,u),
\quad
\tilde u=U(t,x,u),
\quad
\tilde f=\Phi(t,x,u,f),
\end{gather*}
that are projectable to the variable space $(t,x,u)$ and
map each equation from the class~$\mathcal L|_{\mathcal S}$ to an~equation from the same class.
\end{definition}

Each equivalence transformation~$\mathcal T\in G^\sim$ generates a~family
of admissible transformations
$\{(f,\mathcal T|_{(t,x,u)},\mathcal T f)\mid f\in S\}\subset\mathcal G^\sim$,
where $\mathcal T|_{(t,x,u)}$ denotes the restriction of~$\mathcal T$ to the space~$(t,x,u)$.

\begin{definition}
A class of differential equations~$\mathcal L|_{\mathcal S}$ is called \emph{normalized}
if its equivalence groupoid~$\mathcal G^\sim$ is generated by its equivalence group~$G^\sim$,
meaning that for each triple $(f,\varphi,\tilde f)$ from~$\mathcal G^\sim$
there exists a~transformation~$\mathcal T$ from~$G^\sim$
such that $\tilde f =\mathcal T f$ and $\varphi=\mathcal T|_{(t,x,u)}$.
\end{definition}


In the case of single dependent variable we might also consider contact transformations, see, e.g., \cite{camp1903a}.
They are of the form
$\tilde t=T(t,x,u,u_t,u_x)$,
$\tilde x=X(t,x,u,u_t,u_x)$,
$\tilde u=U(t,x,u,u_t,u_x)$,
$\tilde u_{\tilde t}=U^t(t,x,u,u_t,u_x)$,
$\tilde u_{\tilde x}=U^x(t,x,u,u_t,u_x)$,
where the transformation components satisfy the nondegenerate assumption
(i.e., the corresponding Jacobi matrix is nondegenerate)
and the contact condition (meaning the consistence with the contact structure),
and thus the components~$U^t$ and~$U^x$ are defined via the chain rule.
At the same time, any contact transformation between any two second-order evolution equations
that are linear with respect to the second derivative is induced by a~specific point transformation,
cf.~\cite[Proposition 2]{popo2008d} and~\cite[Section~2]{vane2014b}.
This is why all contact transformations between equations from the class~\eqref{GBE}
are exhausted by the prolongations of point admissible transformations,
 and thus the problem of describing contact transformations in the class~\eqref{GBE}
is reduced to that for point transformations.

\section{Equivalence groupoid and equivalence group}\label{GROUPandALGEBRAsection}

It is possible to find not only the equivalence group and the equivalence algebra of the class~\eqref{GBE}
but also its entire equivalence groupoid.
In an implicit form, this groupoid was described in the pioneer paper~\cite{king1991c},
which is, actually,
the first study of the set of admissible transformations in a class of differential equations.
Earlier, in~\cite{cate1989a}, a~conformal transformation was found between equations
of the form~\eqref{GBE} with $f_x=0$.

\begin{theorem}\label{GBEgroupoid}
The class~\eqref{GBE} is normalized in the usual sense.
The usual equivalence group~$G^\sim$ of the class~\eqref{GBE} consists of the transformations
\begin{gather}
\tilde t=\frac{\alpha t+\beta}{\gamma t+\delta},
\qquad
\tilde x=\frac{\kappa x+\mu_1t+\mu_0}{\gamma t+\delta},
\qquad
\tilde{u}=\frac{\kappa(\gamma t+\delta)u-\kappa\gamma x+\mu_1\delta-\mu_0\gamma}
{\alpha\delta-\beta\gamma},
\label{GBE_G}
\\
\tilde f=\frac{\kappa^2}{\alpha\delta-\beta\gamma}f,
\label{GBE_G_f}
\end{gather}
where $\alpha$, $\beta$, $\gamma$, $\delta$, $\mu_0$, $\mu_1$ and $\kappa$ are arbitrary constants
that are defined up to a~nonzero multiplier,
$\alpha\delta-\beta\gamma\ne0$ and $\kappa\ne0$.
\end{theorem}
\begin{proof}
We fix any two equations from the class~\eqref{GBE},
$\mathcal L_f$: $u_t+uu_x+f(t,x)u_{xx}=0$ and
\smash{$\mathcal L_{\tilde f}$}: $\tilde u_{\tilde t}+\tilde u\tilde u_{\tilde x}+\tilde f(\tilde t,\tilde x)\tilde u_{\tilde x\tilde x}=0$,
and find all point transformations, which are of the form~\eqref{TXU}, between these two equations.
For this purpose, we substitute all the tilded variables and derivatives with their expressions in terms of untilded values in $\mathcal L_{\tilde f}$,
including
\begin{gather*}
\tilde u_{\tilde t}=\frac1{{\rm D}_t T}\left({\rm D}_t U-\frac{{\rm D}_x U {\rm D}_t X}{{\rm D}_x X}\right),
\qquad
\tilde u_{\tilde x}=\frac{{\rm D}_x U}{{\rm D}_x X},
\qquad
\tilde u_{\tilde x \tilde x}={\rm D}_x\left(\frac{{\rm D}_x U}{{\rm D}_x X}\right),
\end{gather*}
where
${\rm D}_t=\p_t+u_t\p_u+u_{tt}\p_{u_t}+u_{tx}\p_{u_x}+\dots$ and
${\rm D}_x=\p_x+u_x\p_u+u_{tx}\p_{u_t}+u_{xx}\p_{u_x}+\dots$
are the operators of total derivatives with respect to~$t$ and~$x$.
The obtained equation should be satisfied by all solutions of~$\mathcal L_f$.
This is why the equality derived by substituting $u_t$ with $-uu_x-fu_{xx}$ in view of~$\mathcal L_f$
can be split with respect to~$u_x$ and~$u_{xx}$,
which gives the system of determining equations
for the transformation components~$T$, $X$ and~$U$.

The computation can be simplified by taking into account
the specific structure of generalized Burgers equations,
which are second-order quasilinear evolution equations,
i.e., they are linear with respect to the derivative~$u_{xx}$.%
\footnote{%
Throughout the paper, a~linear function means a~polynomial of degree one or zero.
}
In view of~\cite[Lemma 1]{ivan2010a},
each point transformation between the equations~$\mathcal L_f$ and~$\mathcal L_{\tilde f}$
is projectable both on the space of~$t$ and on the space of~$(t,x)$,
$
\tilde t=T(t),
$ $
\tilde x=X(t,x),
$ $
\tilde u=U(t,x,u),
$
where $T_tX_xU_u\ne0$.%
\footnote{%
We can also use results of~\cite{poch2013d} for the narrower superclass
$u_t+F(t,x,u)u_{xx}+H^1(t,x,u)u_x+H^0(t,x,u)=0$,
which gives more constraints for the transformation components, 
but the property of double projectability itself sufficiently simplifies the computation.
}
The determining equations can be reduced to
\begin{gather*}
U_{uu}=0,
\qquad
\text{and hence}
\qquad
U=U^1(t,x)u+U^0(t,x),
\\
\tilde f=\frac{X_x^2}{T_t}f,
\quad
U^1=\frac{X_x}{T_t},
\quad
U^0=\frac{X_t}{T_t},
\quad
X_{xx}=0,
\quad
U^1_t+U^0_x=0,
\quad
U^0_t=0.
\end{gather*}
The equivalence groupoid~$\mathcal G^\sim$ of the class~\eqref{GBE} is established after solving these determining equations.
The elements of $\mathcal G^\sim$ are defined by~\eqref{GBE_G}, where the initial and the resulting values of the arbitrary element
are connected by~\eqref{GBE_G_f}.
Each transformation of the form~\eqref{GBE_G} maps any equation from the class~\eqref{GBE} to an equation from the same class,
and its prolongation to the arbitrary element~$f$, which is given by~\eqref{GBE_G_f},
is a~point transformation in the joint space $(t,x,u,f)$.
Hence such prolongations of the transformations of the form~\eqref{GBE_G} according to~\eqref{GBE_G_f} constitute
the equivalence group~$G^\sim$ of the class~\eqref{GBE}.
Since any element of~$\mathcal G^\sim$ is induced by an equivalence transformation,
this class is normalized.
\end{proof}

The connected component of the identity transformation in~$G^\sim$ is singled out
by the inequalities $\alpha\delta-\beta\gamma>0$ and $\kappa>0$.
Up to composing with each other and with continuous equivalence transformations,
discrete equivalence transformations are exhausted
by alternating signs in the tuples $(t,u,f)$ and $(x,u)$.

The Lie algebra~$\mathfrak{g}^\sim=
\langle\tilde P^t,\tilde P^x,\tilde D^t,\tilde D^x,\tilde G,\tilde\Pi\rangle$
that corresponds to the equivalence group~$G^\sim$ is called the \emph{equivalence algebra} of the class~\eqref{GBE}.
Its basis elements may be chosen as
\begin{gather*}
\tilde P^t=\p_t,
\qquad
\tilde P^x=\p_x,
\qquad
\tilde D^t=t\p_t-u\p_u-f\p_f,
\qquad
\tilde D^x=x\p_x+u\p_u+2f\p_f,
\\
\tilde G=t\p_x+\p_u,
\qquad
\tilde\Pi=t^2\p_t+tx\p_x+(x-tu)\p_u.
\end{gather*}
The elementary one-parameter transformations from the group~$G^\sim$ that correspond to these basis elements are
\begin{gather}\label{elemG1}
\arraycolsep=10pt\renewcommand{\arraystretch}{1.4}
\begin{array}{rllll}
\hat P^t(\beta)\colon & \tilde t=t+\beta, & \tilde x=x, & \tilde u=u, & \tilde f=f,
\\
\hat P^x(\mu_0)\colon & \tilde t=t, & \tilde x=x+\mu_0, & \tilde u=u, & \tilde f=f,
\\
\hat D^t(\alpha)\colon & \tilde t=\alpha t, & \tilde x=x, & \tilde u=\dfrac 1\alpha u, & \tilde f=\dfrac 1\alpha f,
\\
\hat D^x(\kappa)\colon & \tilde t=t, & \tilde x=\kappa x, & \tilde u=\kappa u, & \tilde f=\kappa^2 f,
\\
\hat G(\mu_1)\colon & \tilde t=t, & \tilde x=x+\mu_1t, & \tilde u=u+\mu_1, & \tilde f=f,
\\
\hat \Pi(\gamma)\colon & \tilde t=\dfrac t{\gamma t+1}, & \tilde x=\dfrac x{\gamma t+1}, & \tilde u=(\gamma t+1)u-\gamma x, & \tilde f=f.
\end{array}
\end{gather}
These transformations are deduced from~\eqref{GBE_G}--\eqref{GBE_G_f} by setting
all the constants (except one, which is present in a~transformation) with the values
that correspond to the identity transformation,
\begin{gather*}
\alpha=1,
\quad
\beta=0,
\quad
\gamma=0,
\quad
\delta=1,
\quad
\kappa=1,
\quad
\mu_1=0,
\quad
\mu_0=0.
\end{gather*}

The equivalence algebra~$\mathfrak g^\sim$ can be found directly using the infinitesimal Lie method,
in a~similar way as finding Lie symmetries of single systems of differential equations~\cite{akha1989a,ovsy1982a}.
In fact, this is not needed since, knowing the complete equivalence group~$G^\sim$,
we can construct the algebra~$\mathfrak g^\sim$ as the set of infinitesimal generators of one-parameter subgroups
of~$G^\sim$, cf.~\cite{kuru2016a}.

The projection of~$\mathfrak g^\sim$ to the space $(t,x,u)$ is the algebra~$\mathfrak g=\langle
P^t,P^x,D^t,D^x,G,\Pi\rangle\simeq\mathfrak g^\sim$ with
\begin{gather*}
P^t=\p_t,
\qquad
P^x=\p_x,
\qquad
D^t=t\p_t-u\p_u,
\qquad
D^x=x\p_x+u\p_u,
\\
G=t\p_x+\p_u,
\qquad
\Pi=t^2\p_t+tx\p_x+(x-tu)\p_u.
\end{gather*}
Both the algebras~$\mathfrak g^\sim$ and~$\mathfrak g$ are realizations
of the so-called reduced (i.e., centerless) full Galilei algebra~\cite{fush1991g}
with space dimension one, 
which is isomorphic to the affine Lie algebra ${\rm aff}(2,\mathbb R)$~\cite{gaze1992a}.
The nonzero commutation relations of~$\mathfrak g$ are
\begin{gather*}
[P^t,D^t]=P^t,
\qquad
[D^t,\Pi]=\Pi,
\qquad
[P^t,\Pi]=2D^t+D^x,
\\
[P^x,D^x]=P^x,
\qquad
[P^x,\Pi]=G,
\qquad
[P^t,G]=P^x,
\qquad
[D^t,G]=G,
\qquad
[G,D^x]=G.
\end{gather*}

The Levi decomposition of the algebra~$\mathfrak g$ is $\mathfrak g=\langle P^t,D^t+\frac12D^x,\Pi\rangle\lsemioplus\langle D^x,P^x,G\rangle$.
Here the subalgebra $\mathfrak f=\langle P^t,D^t+\frac12D^x,\Pi\rangle$ is a~Levi factor of~$\mathfrak g$,
which is a~realization of the algebra ${\rm sl}(2,\mathbb R)$.
The radical~$\mathfrak r=\langle D^x,P^x,G\rangle$ of~$\mathfrak g$ is a~realization of the algebra $A_{3,3}$
from the Mubarakzyanov's list of low-dimensional real algebras~\cite{muba1963b}
(see also~\cite{popo2003a}), which is the almost abelian algebra associated with the $2\times2$ identity matrix.
More specifically, $\mathfrak r=\mathfrak c\lsemioplus\mathfrak n$,
where $\mathfrak n=\langle P^x,G\rangle$ is the nilradical (as well as the maximal abelian ideal)
of both~$\mathfrak r$ and~$\mathfrak g$,
and the span $\mathfrak c=\langle D^x\rangle$ is a~Cartan subalgebra of~$\mathfrak r$.
By~$\mathop{{\rm pr}_{\mathfrak f}}$ and~$\mathop{{\rm pr}_{\mathfrak c}}$ we denote the projectors defined by
the decomposition $\mathfrak g=\mathfrak f\lsemioplus(\mathfrak c\lsemioplus\mathfrak n)$.

As the class~\eqref{GBE} is normalized,
the algebra~$\mathfrak g$ contains the union of Lie invariance algebras 
of all equations~$\mathcal L_f$ from the class~\eqref{GBE}.
Moreover, the algebra~$\mathfrak g$ appears to coincide with this union, cf.\ Section~\ref{LieSYMsection},
which displays the strong normalization of the class~\eqref{GBE} in the infinitesimal sense.

\section{Lie symmetries}\label{LieSYMsection}

Classical symmetry analysis of the classical Burgers equation and its generalizations related to the class~\eqref{GBE}
has been carried out since the 1960s.
The maximal Lie invariance group of the Burgers equation
was computed by Katkov~\cite{katk1965a} in the course of group classification of equations of the general form $u_t+uu_x=(f(u)u_x)_x$.
The maximal Lie invariance algebra of the Burgers equation is five-dimensional and is spanned by the vector fields
$P^x$, $G$, $P^t$, $D^t+\frac12D^x$ and $\Pi$, cf.~the algebra~$\mathfrak g^5$ in Table~\ref{TableSubalgebras} below.

The group classification problem for the subclass of~\eqref{GBE} singled out by the constraint $f_x=0$
was considered in~\cite{doyl1990a,wafo2004d} without proper use of equivalence transformations.
As a~result, the classification lists presented there contain needless cases,
which was already remarked in~\cite{king1991c} concerning~\cite{doyl1990a};
see also footnote~4 in~\cite{poch2014a}.
The same remark is true for the group classification of equations
of the general form $u_t+g(t,x)uu_x+f(t,x)u_{xx}=0$ with $fg\ne0$
in~\cite{qu1995b}.
Moreover, classification cases were unnecessarily split into several subcases
depending on the structure of symmetry algebras therein,
and, as far as it can be analyzed, some classification cases were missed, in particular,
due to over-gauging parameters, which is not allowed.

\noprint{\footnote{%
Compare the results of~\cite{qu1995b} with our results,
cf.\ Table~\ref{TableSubalgebras}.
For one-dimensional subalgebras~\cite[Table 3]{qu1995b}
it is unnecessary to separate cases 5 and 6.
In the list of two-dimensional subalgebras~\cite[Table 2]{qu1995b}
case 6 should correspond the subalgebra $\mathfrak g^{2.5}$ in Table~\ref{TableSubalgebras},
and case 7 seems to be related to $\mathfrak g^{2.6}_a$. 
Moreover, constant multipliers are necessary in cases 3 and 4.
For three-dimensional subalgebras~\cite[Table 1]{qu1995b}
it is unnecessary to separate cases 4 and 5,
as well as cases 7, 8 and 9.
More detailed analysis looks rather complicate due to
and mistakes in numeration.
}}

\subsection{Determining equations for Lie symmetries}\label{CSdeteqSECTION}

A vector field that generates a~one-parameter Lie symmetry group
of an equation~$\mathcal L_f$: $L^f[u]=0$ from the class~\eqref{GBE}
is of the form $Q=\tau(t,x,u)\p_t+\xi(t,x,u)\p_x+\eta(t,x,u)\p_u$ 
and satisfies the infinitesimal invariance criterion
\begin{gather}\label{CSICI}
Q_{(2)}L^f[u]\,\big|_{\mathcal L_f}
\equiv\big(\eta^t+\eta u_x+\eta^xu+\xi f_x u_{xx}+\tau f_t u_{xx}+\eta^{xx}f\big)\big|_{\mathcal L_f}=0,
\end{gather}
where $Q_{(2)}$ is the usual second-order prolongation of~$Q$~\cite{olve1993b,ovsy1982a}, and
$\eta^t$, $\eta^x$, $\eta^{xx}$ are prolongation components, which are computed by
$\eta^t={\rm D}_t\eta-u_t{\rm D}_t\tau-u_x{\rm D}_t\xi$,
$\eta^x={\rm D}_x\eta-u_t{\rm D}_x\tau-u_x{\rm D}_x\xi$ and
$\eta^{xx}={\rm D}_x^2\eta-u_t{\rm D}_x^2\tau-2u_{tx}{\rm D}_x\tau-u_x{\rm D}_x^2\xi-2u_{xx}{\rm D}_x\xi$.
Using the restriction $L^f[u]=0$, we substitute $u_t=-uu_x-fu_{xx}$ for~$u_t$
and then split the result with respect to $u_{tx}$, $u_{xx}$, $u_{x}$ and~$u$.
After simplifying we obtain a~system of determining equations on the components~$\tau$, $\xi$ and~$\eta$,
\begin{gather}
\label{CSdeteq1}
\tau_x=0,
\quad
\tau_u=0,
\quad
\xi_u=0,
\quad
\eta_{uu}=0,
\quad
\eta^1_x=0,
\\
\label{CSdeteq2}
\eta^0-\xi_t=0,
\quad
\eta^1+\tau_t-\xi_x=0,
\quad
\eta^1_t=-\eta^0_x,
\quad
\eta^0_t=0,
\\
\label{CSCC}
\tau f_t+\xi f_x+(\tau_t-2\xi_x)f=0.
\end{gather}
Equations~\eqref{CSdeteq1} imply $\tau=\tau(t)$, $\xi=\xi(t,x)$ and $\eta=\eta^1(t)u+\eta^0(t,x)$.
Then making use of equations~\eqref{CSdeteq2} we specify the form of the components of~$Q$,
\begin{gather}
\label{CStauxieta}
\tau=c_2t^2+c_1t+c_0,
\quad
\xi=(c_2t+c_3)x+c_4t+c_5,
\quad
\eta=(-c_2t+c_3-c_1)u+c_2x+c_4,
\end{gather}
where $c_0$, \dots, $c_5$ are arbitrary constants.
In view of results of Section~\ref{GROUPandALGEBRAsection},
we could postulate the form~\eqref{CStauxieta} for the components of Lie symmetry vector fields
of equations from the class~\eqref{GBE} from the very beginning.
Indeed, the normalization of the class~\eqref{GBE} proved in Theorem~\ref{GBEgroupoid} implies
that the maximal Lie invariance algebra~$\mathfrak g_f$ of any equation~$\mathcal L_f$ from the class~\eqref{GBE}
is contained by the algebra~$\mathfrak g$,
and the components of any vector field from~$\mathfrak g$ are of the form~\eqref{CStauxieta}.
Moreover, for any constant tuple $(c_0,\dots,c_5)$ the equation~\eqref{CSCC} has a~nonzero solution for~$f$.
This means that each element of the algebra~$\mathfrak g$ is a~Lie symmetry of an equation from the class~\eqref{GBE},
i.e.\ $\mathfrak g=\bigcup_{f\in\mathcal S}\mathfrak g_f$.

The equation~\eqref{CSCC} is the only \emph{classifying condition} for Lie symmetries of equations from the class~\eqref{GBE}.
Depending on values of the arbitrary element~$f$,
the classifying condition imposes  additional constraints for the constants $c_0$,~\dots, $c_5$.
Varying~$f$ and splitting~\eqref{CSCC} with respect to~$f$ and its derivatives, we get $c_0=\dots=c_5=0$.
Therefore, the \emph{kernel invariance algebra} of the class~\eqref{GBE},
i.e.\ the intersection of the maximal Lie invariance algebras of equations from this class, is $\{0\}$.

\subsection{Appropriate subalgebras}\label{SubalgebrasSection}

As the class~\eqref{GBE} is normalized (see Theorem~\ref{GBEgroupoid})
and its equivalence algebra~$\mathfrak g^\sim$ is finite-dimensional,
it is convenient to carry out its group classification using the algebraic method. 
Although one could solve the group classification problem for 
the class~\eqref{GBE} using the direct method, 
the algebraic method is much more effective on both the steps 
of computing and arranging classification cases, 
in particular, checking their inequivalence; cf.\ \cite[p.~3]{vane2009a}.
Recall that the normalization of the class~\eqref{GBE} has two consequences:
\begin{itemize}
 \item The maximal Lie invariance algebra of every
equation from this class is contained in the projection~$\mathfrak g$ of the equivalence algebra~$\mathfrak g^\sim$
to the space $(t,x,u)$.
 \item Equations~$\mathcal L_f$ and~$\mathcal L_{\tilde f}$ from the class~\eqref{GBE} are similar with
respect to point transformations if and only if they are $G^\sim$-equivalent. 
\end{itemize}
Therefore, in order to obtain the exhaustive group classification of the class~\eqref{GBE}, it suffices to construct a~list of
inequivalent appropriate subalgebras of~$\mathfrak g$
and then to find the corresponding values of the arbitrary element~$f$ for each subalgebra from the list.
We call a~subalgebra~$\mathfrak s\subset\mathfrak g$ \emph{appropriate} if~$\mathfrak s$ is the maximal Lie
invariance algebra of an equation~$\mathcal L_f$ from the class~\eqref{GBE}.
In other words, a~subalgebra~$\mathfrak s$ of~$\mathfrak g$ is appropriate if
there exists a~value~$f^0$ of the arbitrary element~$f$ such that the following conditions hold:
\begin{enumerate}
\item
The components of every $Q\in\mathfrak s$ satisfy the classifying condition~\eqref{CSCC} with $f=f^0$
or, equivalently, $f^0$ is an invariant of a~subalgebra $\tilde{\mathfrak s}\subset\mathfrak g^\sim$
whose projection to the space of $(t,x,u)$ coincides with~$\mathfrak s$.
\item
The algebra~$\mathfrak s$ is maximal among the Lie invariance algebras of the equation~$\mathcal L_{f^0}$.
\end{enumerate}

We classify appropriate subalgebras of the algebra~$\mathfrak g$
up to the equivalence relation generated by the
adjoint action of the group~$G^\sim$ on~$\mathfrak g$.%
\footnote{\looseness=-1%
The subalgebras of an algebra isomorphic to~${\rm aff}(2,\mathbb R)$ 
were first classified in~\cite{gaze1992a}
with respect to the group of internal automorphisms.
Parameterized families of inequivalent subalgebras
were additionally partitioned depending on their algebraic structure,
which is unnecessary for the group classification of the class~\eqref{GBE}.
Moreover, the list of subalgebras obtained therein is large.
It consists of 44 families of subalgebras,
and most of them are not appropriate as maximal Lie invariance algebras
of equations from the class~\eqref{GBE}.
Since the proof of the classification was not presented,
it is impossible to check its correctness at a~glance.
This is why we classify appropriate subalgebras of~$\mathfrak g$ independently, 
without using the results of~\cite{gaze1992a},
which is much easier than the classification of all subalgebras of~$\mathfrak g$.
The list of inequivalent (nonzero) appropriate subalgebras presented in Table~\ref{TableSubalgebras} below
includes only 19 families of subalgebras.%
}
See, e.g., \cite[Chapter 3.3]{olve1993b}, \cite[Section~14.7]{ovsy1982a}
or \cite{bihl2012b,bihl2009b} for relevant elementary techniques
and \cite{pate1975a} for more sophisticated methods.
The radical~$\mathfrak r$ and the nilradical~$\mathfrak n$
are megaideals (i.e., fully characteristic ideals) of~$\mathfrak g$
and hence they are $G^\sim$-invariant.
To characterize classification cases, with any subalgebra~$\mathfrak s$ of~$\mathfrak g$
we associate the $G^\sim$-invariant values
$\dim\mathfrak s\cap\mathfrak r$, $\dim\mathfrak s\cap\mathfrak n$,
$\dim\mathop{{\rm pr}_{\mathfrak f}}\mathfrak s$ and $\dim\mathop{{\rm pr}_{\mathfrak c}}\mathfrak s$.
The adjoint actions of the elementary equivalence transformations~\eqref{elemG1}
on the basis vector fields of~$\mathfrak g$ are as follows:

\vspace*{-2.4mm}

{\footnotesize
\begin{gather*}\renewcommand{\arraystretch}{1.6}
\begin{array}{r|cccccc}
\mathrm{Ad}     &P^t                       &P^x           &D^t            &D^x            &G            &\Pi                                   \\\hline
\hat P^t(\beta) &P^t                         &P^x           &D^t{-}\beta P^t&D^x            &G{-}\beta P^x&\Pi{-}\beta(2D^t{+}D^x){+}\beta^2P^t\\
\hat P^x(\mu_0) &P^t                           &P^x           &D^t            &D^x{-}\mu_0 P^x&G            &\Pi{-}\mu_0 G                     \\
\hat D^t(\alpha)&\alpha P^t                      &P^x           &D^t            &D^x            &\alpha^{-1}G &\alpha^{-1}\Pi                  \\
\hat D^x(\kappa)&P^t                               &\kappa P^x    &D^t            &D^x            &\kappa G     &\Pi                           \\
\hat G(\mu_1)   &P^t{+}\mu_1 P^x                     &P^x           &D^t{+}\mu_1 G  &D^x{-}\mu_1 G  &G            &\Pi                         \\
\hat \Pi(\gamma)&P^t{-}\gamma(2D^t{+}D^x){+}\gamma^2\Pi&P^x{-}\gamma G&D^t{-}\gamma\Pi&D^x            &G            &\Pi                       \\
\end{array}
\end{gather*}
}

To efficiently recognize inequivalent appropriate subalgebras of~$\mathfrak g$,
we consider their projections on the Levi factor~$\mathfrak f$.
These projections are necessarily subalgebras of~$\mathfrak f$,
and, moreover, the projections of equivalent subalgebras of~$\mathfrak g$ are equivalent as subalgebras of~$\mathfrak f$.
A~complete list of inequivalent subalgebras of the algebra ${\rm sl}(2,\mathbb R)$ is well known.
In terms of the realization~$\mathfrak f$,
it is exhausted by $\{0\}$, $\langle P^t\rangle$, $\langle D^t+\frac12D^x\rangle$, $\langle P^t+\Pi\rangle$,
$\langle P^t, D^t+\frac12D^x \rangle$ and~$\mathfrak f$ itself.
Considering each of the listed subalgebras of~$\mathfrak f$
as a~projection of an appropriate subalgebra,
we try to add elements of the radical~$\mathfrak r$ to the basis elements of this subalgebra,
and to additionally extend the basis by elements from the radical~$\mathfrak r$.

Some properties of appropriate subalgebras of~$\mathfrak g$ directly follow
from the classifying condition~\eqref{CSCC}.
Below by~$\mathfrak s$ we denote an appropriate subalgebra of~$\mathfrak g$.

\begin{lemma}\label{LemmaOnPxG}
If $\mathfrak s\cap\mathfrak n\ne\{0\}$, then $\mathfrak s\cap\mathfrak r=\mathfrak n$.
\end{lemma}

\begin{proof}
The condition $\mathfrak s\cap\mathfrak n\ne\{0\}$ implies
that the algebra~$\mathfrak s$ contains a~vector field $aG+bP^x$ with constants $(a,b)\ne(0,0)$.
Substituting the values
$\tau=0$ and $\xi=at+b$ corresponding to $aG+bP^x$ into~\eqref{CSCC}, we get $f_x=0$.
For such~$f$, both the pairs $(\tau,\xi)=(0,1)$ and $(\tau,\xi)=(0,t)$ solve~\eqref{CSCC},
and the pair $(\tau,\xi)=(0,x)$ is not a~solution since $f\ne0$.
Therefore, the algebra~$\mathfrak s$ contains both~$P^x$ and~$G$ and does not contain $D^x$.
\end{proof}

\begin{corollary}
$\mathfrak s\cap\mathfrak r\in\big\{\{0\},\,\mathfrak c,\,\mathfrak n\big\}$.
\end{corollary}

\begin{lemma}\label{LemmaOnDx}
If $\mathfrak s\cap\mathfrak r=\mathfrak c$, then
$\mathfrak s\subset\mathfrak f\oplus\mathfrak c$ and $\dim\mathfrak s\leqslant2$.
\end{lemma}

\begin{proof}
Suppose that $\mathfrak s\nsubseteq\mathfrak f\oplus\mathfrak c$
and thus $\mathfrak s\setminus (\mathfrak f\oplus\mathfrak c)\ne\varnothing$.
Each element of this set difference is of the form
$Q=a_0P^t+a_1D^t+a_2\Pi+a_3D^x+a_4G+a_5P^x$ with $(a_4,a_5)\ne(0,0)$.
Since $D^x,Q\in\mathfrak s$, we have $[Q,D^x]=a_4G+a_5P^x\in\mathfrak s$.
Then Lemma~\ref{LemmaOnPxG} implies that $D^x\not\in\mathfrak s$, which gives a~contradiction.
Therefore, $\mathfrak s\subset\mathfrak f\oplus\mathfrak c$.

If $\dim\mathfrak s>2$, then $\dim\mathfrak s\cap\mathfrak f\geqslant2$ 
and hence, modulo $G^\sim$-equivalence, $\mathfrak s\supseteq\langle P^t,D^t+\frac12D^x,D^x\rangle$.
Therefore, we also get $\mathfrak s\supseteq\langle P^t,D^t\rangle$.
The substitution of the values $(\tau,\xi)=(1,0)$ and $(\tau,\xi)=(t,0)$
corresponding to~$P^t$ and~$D^t$, respectively,
into~\eqref{CSCC} leads to a~system on~$f$ that is inconsistent with the
constraint $f\ne0$ for the arbitrary element~$f$.
The obtained contradiction means that $\dim\mathfrak s\leqslant2$
and hence $\dim\mathfrak s\cap\mathfrak f\leqslant1$.
\end{proof}

\begin{corollary}
If $\mathfrak s\cap\mathfrak r=\mathfrak c$, then
$\mathfrak s\in\big\{\langle D^x\rangle, \langle D^x,P^t\rangle, \langle D^x,D^t+\frac12D^x\rangle, \langle D^x,P^t+\Pi\rangle\big\}\bmod G^\sim$.
\end{corollary}

This gives the subalgebras~$\mathfrak g^{1.1}$, $\mathfrak g^{2.2}$--$\mathfrak g^{2.4}$ of Table~1, respectively.

\begin{lemma}\label{LemmaOnSubalgerasWithPxG}
If $\mathfrak s\cap\mathfrak r=\mathfrak n$, then
either $\dim\mathfrak s\leqslant3$ and $\mathfrak s\cap\mathfrak f=\{0\}$
or $\mathfrak s=\mathfrak f\lsemioplus\mathfrak n$.
\end{lemma}

\begin{proof}
Suppose that $\dim\mathfrak s>3$. Then $\dim\mathop{{\rm pr}_{\mathfrak f}}\mathfrak s\geqslant2$,
i.e., modulo~$G^\sim$, the algebra~$\mathfrak s$ contains
the vector fields $Q^1=P^t+aD^x$ and $Q^2=D^t+bD^x$ with some constants~$a$ and~$b$.
Therefore, the commutator $[Q^1,Q^2]=P^t$ also belongs to~$\mathfrak s$,
and the classifying condition~\eqref{CSCC} in view of involving $P^x$ and~$P^t$ implies that $f=\const$.

In the same way, the condition $f=\const$ is derived in the case $\mathfrak s\cap\mathfrak f\ne\{0\}$
since after substituting the components of~$P^x$ and of any nonzero element of~$\mathfrak f$ into~\eqref{CSCC}
we obtain the equations $f_x=0$ and $f_t=0$.

The maximal Lie invariance algebra of the equation~$\mathcal L_f$ with $f=\const$,
which is the classical Burgers equation,
is the five-dimensional algebra
$\mathfrak s=\mathfrak f\lsemioplus\mathfrak n$.
\end{proof}

\begin{corollary}\label{CorellaryOnSubalgerasWithPxG}
If $\mathfrak s\cap\mathfrak r=\mathfrak n$, then
$\mathfrak s=\mathfrak s_{\mathfrak f}\lsemioplus\mathfrak n$, where
\begin{gather*}
\mathfrak s_{\mathfrak f}\in\big\{\{0\},\, \langle P^t+\tfrac12D^x\rangle,\, \langle D^t+aD^x\rangle,\, \langle P^t+\Pi+aD^x\rangle,\, \mathfrak f\big\}\bmod G^\sim,
\end{gather*}
the parameter~$a$ runs through $\mathbb R\setminus\{0\}$, and $a>0\bmod G^\sim$.
\end{corollary}

The list presented in Corollary~\ref{CorellaryOnSubalgerasWithPxG} gives
the subalgebras~$\mathfrak g^{2.1}$, $\mathfrak g^{3.1}$--$\mathfrak g^{3.3}$ and $\mathfrak g^5$ of Table~1, respectively.

\begin{corollary}
The dimension of any appropriate subalgebra of~$\mathfrak g$ is not greater than $5$.
\end{corollary}

Below we consider the last case for $\mathfrak s\cap\mathfrak r$, $\mathfrak s\cap\mathfrak r=\{0\}$.
Then we obviously have $\dim\mathfrak s\leqslant3$.
In fact, the upper bound for $\dim\mathfrak s$ can be lowered.

\begin{lemma}\label{LemmaOnSubalgerasWith0PartOfRadical}
If $\mathfrak s\cap\mathfrak r=\{0\}$, then $\dim\mathfrak s\leqslant2$.
Moreover, if additionally $\dim\mathfrak s=2$, then
$\mathop{{\rm pr}_{\mathfrak c}}\mathfrak s=\mathfrak c$ and
$\mathfrak s\ne\langle P^t,D^t\rangle\bmod G^\sim$.
\end{lemma}

\begin{proof}
Suppose that $\dim\mathfrak s\geqslant2$ and $\mathop{{\rm pr}_{\mathfrak c}}\mathfrak s=\{0\}$.
Modulo $G^\sim$-equivalence, we can assume that
$\mathop{{\rm pr}_{\mathfrak f}}\mathfrak s\supseteq\langle P^t,D^t+\frac12D^x\rangle$.
In view of the classifying condition~\eqref{CSCC},
the invariance of~$\mathcal L_f$ with respect to~$\mathfrak s$ then implies $f=\const$.
Recall that the maximal Lie invariance algebra of the equation~$\mathcal L_f$
for any (nonzero) constant~$f$ contains~$\mathfrak n$,
which contradicts the condition $\mathfrak s\cap\mathfrak r=\{0\}$.
Therefore, $\mathop{{\rm pr}_{\mathfrak c}}\mathfrak s=\mathfrak c$ if $\dim\mathfrak s\geqslant2$.

Suppose that $\dim\mathfrak s=3$.
Therefore, $\mathop{{\rm pr}_{\mathfrak f}}\mathfrak s=\mathfrak f$
and hence $\mathfrak s\simeq{\rm sl}(2,\mathbb R)$,
i.e.\ $\mathfrak s$ is a~Levi factor of~$\mathfrak g$.
Then the Levi--Malcev theorem (or the direct computation of commutation relations of~$\mathfrak s$)
implies $\mathop{{\rm pr}_{\mathfrak c}}\mathfrak s=\{0\}$.
This contradicts the above conclusion that
$\mathop{{\rm pr}_{\mathfrak c}}\mathfrak s=\mathfrak c$ if $\dim\mathfrak s\geqslant2$.

Similarly to Lemma~\ref{LemmaOnDx},
the condition $\mathfrak s\supseteq\langle P^t,D^t\rangle$
implies $f=0$, which contradicts the original inequality $f\ne0$ for the arbitrary element~$f$.
\end{proof}

If $\dim\mathfrak s=1$, then
$\mathfrak s=\langle Q_{\mathfrak f}+a_3D^x+a_4G+a_5P^x\rangle$,
where
$Q_{\mathfrak f}\in\big\{P^t, D^t, P^t+\Pi\big\}\bmod G^\sim$
and $a_3$, $a_4$ and~$a_5$ are constants.
Consider $Q_{\mathfrak f}=P^t$.
First suppose that $\mathop{{\rm pr}_{\mathfrak c}}\mathfrak s=\{0\}$ and thus $a_3=0$.
The coefficient~$a_5$ is gauged to zero by $\hat G(-a_5)$,
and $a_4\in\{0,1\}$ up to scaling $\hat D^t(a_4^{-1})$ for $a_4\ne0$.
If $\mathop{{\rm pr}_{\mathfrak c}}\mathfrak s=\mathfrak c$, i.e.\ $a_3\ne0$,
then we scale~$a_3$ to 1 by $\hat D^t(a_3)$ and by a subsequent rescaling of the basis vector field.
Then we set the modified parameters $a_4$ and~$a_5$ to zero by $\hat G(a_4)$ and $\hat P^x(a_4+a_5)$.
The other two $G^\sim$-inequivalent values of~$Q_{\mathfrak f}$ are studied in a~similar way.
For the coefficient~$a_3$, we can then only alternate the sign of~$a_3-\frac12$ or $a_3$, respectively.
The coefficients~$a_4$ and~$a_5$ can be set to zero by $\hat G(\mu_1)$ and $\hat P^x(\mu_0)$
for some $\mu^1$ and~$\mu^0$,
except for $Q_{\mathfrak f}=D^t$ and $a_3\in\{0,1\}$
where the nonzero value~$a_5$ (resp.\ $a_4$) can be only scaled to one if $a_3=0$ (resp.\ if $a_3=1$).
In total, this results in the subalgebras $\mathfrak g^{1.2}$--$\mathfrak g^{1.8}_a$.

In the case $\dim\mathfrak s=2$, up to $G^\sim$-equivalence we have
that the subalgebra~$\mathfrak s$ is spanned by two vector fields of the form
$Q^1=P^t+b_3D^x+b_4G+b_5P^x$ and $Q^2=D^t+a_3D^x+a_4G+a_5P^x$
with some constants $a_3$, $a_4$, $a_5$, $b_3$, $b_4$ and~$b_5$.
The commutation relation for~$\mathfrak s$ is $[Q^1,Q^2]=Q^1$.
Expanding it and collecting the coefficients of~$D^x$, we derive $b^3=0$
and hence $a_3\ne\frac12$ since $\mathop{{\rm pr}_{\mathfrak c}}\mathfrak s=\mathfrak c$.
Then collecting the coefficients of~$G$ and~$P^x$ leads to
the equations $(a_3-2)b_4=0$ and $(a_3-1)b_5+a_4=0$, respectively.
We set $b_5$ to zero by $\hat G(-b_5)$, and hence also $a_4=0$.
The coefficient~$b_4$ is zero if $a_3\ne2$,
and its nonzero value is scaled to one.
For $a_3\ne0$ we can set $a_5=0$ by $\hat P^x(-a_5/a_3)$,
and for $a_3=0$ the nonzero value of $a_5$ is scaled to one.
The simultaneous vanishing $a_3=a_5=0$ is not possible
in view of Lemma~\ref{LemmaOnSubalgerasWith0PartOfRadical}.
Therefore, in this case we obtain the subalgebras $\mathfrak g^{2.5}$--$\mathfrak g^{2.7}$,
which completes the computation of $G^\sim$-inequivalent appropriate subalgebras of~$\mathfrak g$.

\begin{table}[hb!]\footnotesize
\begin{center}
\caption{\footnotesize The group classification of the class~\eqref{GBE}
\strut}\label{TableSubalgebras}
\renewcommand{\arraystretch}{1.95} 
\begin{tabular}{|l|l|c|c|l|}
\hline
\hfil$\mathfrak s\subset\mathfrak g$
                     &\hfil Basis of~$\mathfrak g_f$
                                     &$f(t,x)$                        &$\omega$            &\hfil Constraints
\\
\hline
\hline
$\mathfrak g^{1.1}$  &$D^x$          &$x^2h(\omega)$                  &$t$                 &$\big((\alpha\omega^2+\beta\omega+\gamma)h\big)_{\omega}\ne0$
\\
\hline
$\mathfrak g^{1.2}$  &$P^t$          &$h(\omega)$                     &$x$                 &$(\alpha\omega+\beta)h_{\omega}\ne\gamma h$
\\[0.75ex]
\hline
$\mathfrak g^{1.3}$  &$P^t+G$        &$h(\omega)$                     &$x-\dfrac{t^2}2$    &$(\alpha\omega+\beta)h_{\omega}\ne\gamma h$
\\
\hline
$\mathfrak g^{1.4}$  &$P^t+D^x$      &$e^{2t}h(\omega)$               &$e^{-t}x$           &$h_{\omega}\ne0, \quad \omega h_{\omega}\ne2h$
\\
\hline
$\mathfrak g^{1.5}$  &$D^t+P^x$      &$\dfrac1{t}h(\omega)$           &$x-\ln|t|$          &$h_{\omega}\ne0, \quad h_{\omega}\ne-h$
\\[0.5ex]
\hline
$\mathfrak g^{1.6}_a$&$D^t+aD^x$     &$\dfrac{|t|^{2a}}{t}h(\omega)$  &$|t|^{-a}x$         &see table notes;\quad $a\geqslant\frac12\bmod G^\sim$
\\[0.5ex]
\hline
$\mathfrak g^{1.7}$  &$D^t+D^x+G$    &$t\,h(\omega)$                  &$\dfrac xt-\ln|t|$  &$h_{\omega}\ne0, \quad h_{\omega}\ne-h$                   \\[0.5ex]
\hline
$\mathfrak g^{1.8}_a$&$P^t+\Pi+aD^x$ &$e^{2a\arctan t}h(\omega)$      &$\dfrac{e^{-a\arctan t}}{\sqrt{t^2+1}}x$
                                                                                           &$\omega h_{\omega}\ne2h$;\quad $a\geqslant0\bmod G^\sim$
\\[1ex]
\hline
\hline
$\mathfrak g^{2.1}$  &$P^x,\ G$         &$h(\omega)$                   &$t$                &$(\alpha\omega^2+\beta\omega+\gamma)h_{\omega}\ne\delta h$
\\
\hline
$\mathfrak g^{2.2}$  &$D^x,\ P^t$       &$x^2$                         &                   &
\\
\hline
$\mathfrak g^{2.3}$  &$D^x,\ D^t$       &$\varkappa\dfrac{x^2}{t}$     &                   &$\varkappa\ne0$,\quad $\varkappa>0\bmod G^\sim$
\\[1ex]
\hline
$\mathfrak g^{2.4}$  &$D^x,\ P^t+\Pi$   &$\dfrac{\varkappa x^2}{t^2+1}$&                   &$\varkappa\ne0$,\quad $\varkappa>0\bmod G^\sim$
\\[0.75ex]
\hline
$\mathfrak g^{2.5}$  &$P^t,\ D^t+P^x$   &$e^{-x}$                      &                   &
\\
\hline
$\mathfrak g^{2.6}_a$&$P^t,\ D^t+aD^x$  &$|x|^{2-1/a}$                 &                   &$a\ne0,\frac12$
\\
\hline
$\mathfrak g^{2.7}$  &$P^t+G,\ D^t+2D^x$&$\varkappa\left|x-\dfrac{t^2}{2}\right|^{3/2}$\rule{0pt}{18pt}
                                                                       &                   &$\varkappa\ne0$,\quad $\varkappa>0\bmod G^\sim$
\\[1ex]
\hline
\hline
$\mathfrak g^{3.1}$  &$P^x,\ G,\ P^t+\frac1{2}D^x$&$\e e^t$            &                   &
\\
\hline
$\mathfrak g^{3.2}_a$&$P^x,\ G,\ D^t+aD^x$        &$\e |t|^{2a-1}$     &                   &$a\ne\frac12, \quad a>\frac12\bmod G^\sim$
\\[0.75ex]
\hline
$\mathfrak g^{3.3}_a$&$P^x,\ G,\ P^t+\Pi+a D^x$   &$\e e^{2a\arctan t}$&                   &$a\ne0, \quad a>0\bmod G^\sim$
\\
\hline
\hline
$\mathfrak g^5$      &$P^x,\ G,\ P^t,\ D^t+\frac12D^x,\ \Pi$       &$1$&                   &
\\
\hline
\end{tabular}
\end{center}
Here $\e=\pm1\bmod G^\sim$.
The constants~$a$ and $\varkappa$ and the (nonvanishing) function~$h$
should satisfy constraints in the last column
for the expression of~$f$ to be well defined
and for the corresponding Lie invariance algebra to be maximal.
If possible, we gauge the constants~$a$ and~$\varkappa$ by equivalence transformations,
which is also indicated in the last column.
For the algebras~$\mathfrak g^{1.1}$, $\mathfrak g^{1.2}$, $\mathfrak g^{1.3}$ and~$\mathfrak g^{2.1}$,
the constants $\alpha$, $\beta$, $\gamma$ and $\delta$
involved in the corresponding inequalities for~$h$ are arbitrary but are not simultaneously zeros.
The constraints for~$\mathfrak g^{1.6}_a$ are $h_{\omega}\ne0$ and
\begin{gather*}
a(\omega+\alpha)h_{\omega}\ne(2a-1)h \quad\text{if}\quad (a-2)(a-1)\alpha=0,
\\
(\omega+\beta)h_{\omega}\ne2h \quad\text{if}\quad (a-1)a\beta=0,
\\
(a-1)(\omega+\gamma)h_{\omega}\ne(2a-1)h \quad\text{if}\quad a(a+1)\gamma=0.
\end{gather*}
\end{table}

\subsection{Classification results}

In order to compute the associated values of the arbitrary element~$f$,
for each of the listed inequivalent subalgebras of~$\mathfrak g$
we substitute the components~$\tau$ and $\xi$ of its basis vector fields
into the classifying condition~\eqref{CSCC}
and then solve the obtained system of differential equations on~$f$.
This system does have solutions
and, moreover, at least for a~subset of its solutions
the involved subalgebra of~$\mathfrak g$ is the maximal Lie invariance algebra
for the corresponding equations from the class~\eqref{GBE}.
This means that the collection of properties of appropriate subalgebras
derived in Section~\ref{SubalgebrasSection} gives a~necessary and sufficient condition
for a~subalgebra of~$\mathfrak g$ to be appropriate.

A complete list of inequivalent appropriate subalgebras of~$\mathfrak g$ and the corresponding values
for~$f$ is presented in Table~\ref{TableSubalgebras}.
Since the class~\eqref{GBE} is normalized, the table provides its exhaustive group classification.

\section{Classical similarity solutions}\label{SOLUTIONSsection}

The solution of the group classification problem for a~class of differential equations
can be used for finding exact solutions of equations from the class.
The standard procedure for this purpose
starts with classifying subalgebras of the maximal Lie invariance algebra
of each equation listed in the course of group classification.
Then,
using invariants of obtained inequivalent subalgebras,
one constructs ansatzes for the unknown function
and derives the corresponding reduced equations.
In general, reduced equations are simpler for solving than their original counterparts
since they have less number of independent variables.
The last step of the procedure is to construct at least particular solutions of reduced equations,
which gives exact solutions of the corresponding original equations.

In order to optimize the reduction process for equations from the class~\eqref{GBE}, we exploit two special reduction techniques.

The first technique is available due to the class~\eqref{GBE} is normalized.
Roughly speaking, this technique can be characterized as the classification
of Lie reductions with respect to the equivalence group $G^\sim$ of the whole class,
rather than with respect to the Lie symmetry group of the equation to be reduced.
Thus, this technique  is related to the algebraic method of group classification.
Since the class~\eqref{GBE} is normalized,
the projection of $G^\sim$ to the space $(t,x,u)$ contains the point symmetry groups of all equations from the class~\eqref{GBE},
and hence the maximal Lie invariance algebras of these equations are subalgebras of
the projection~$\mathfrak g$ of the equivalence algebra~$\mathfrak g^\sim$.
Recall that equations~$\mathcal L_f$ and~$\mathcal L_{\tilde f}$ from the class~\eqref{GBE} are similar with
respect to a~point transformation if and only if they are $G^\sim$-equivalent.
The similarity of equations $\mathcal L_f$ and~$\mathcal L_{\tilde f}$
implies the equivalence of their maximal Lie invariance algebras $\mathfrak g_f$ and~$\mathfrak g_{\tilde f}$
and establishes a~one-to-one correspondence between the sets of subalgebras of these algebras.
Subalgebras of $\mathfrak g_f$ and~$\mathfrak g_{\tilde f}$ are obviously subalgebras of the algebra $\mathfrak g$.
So, it suffices to classify inequivalent subalgebras of the algebra~$\mathfrak g$
(cf.\ Section~\ref{SubalgebrasSection}),
that are appropriate for Lie reduction of equations from the class~\eqref{GBE}.
This approach allows us to avoid the separate implementation of the Lie reduction procedure
for each of the nineteen classification cases listed in Table~\ref{TableSubalgebras}.

The second technique, which was systematically used in~\cite{fush1994a,fush1994b} and discussed in~\cite{popo1995c},
is to construct ansatzes in such a~way
that reduced equations are of a~simple and similar form.
Thus, the algebras~$\mathfrak g^{1.0}$ and~$\mathfrak g^{1.1}$ give
trivial first-order ordinary differential equations.
Reduced equations constructed using the algebras \mbox{$\mathfrak g^{1.2}$--$\mathfrak g^{1.8}_a$} are of order two.
For all these algebras, we choose the invariant independent variable~$\omega$ linear in~$x$
with coefficients dependent at most on~$t$,
and the general form of ansatzes is $u=F(t)\varphi(\omega)+G(t,x)$ with $G_{xx}=0$.
Here the intention is to make reduced equations of the same general form~\eqref{RedEqClass}.
After constructing intermediate ansatzes and the correspondent reduced equations,
in some cases it is necessary to change the invariant dependent variable~$\varphi$,
e.g.~$\varphi=\phi+1$,
in order to push all second-order reduced equations into the class~\eqref{RedEqClass}.

Note that classical Lie reductions of equations from the class~\eqref{GBE} were carried out earlier only for the subclass with $f_x=0$
\cite{doyl1990a,wafo2004d} with some weaknesses%
\footnote{%
More specifically, in~\cite{doyl1990a} optimal systems of subalgebras were constructed for the corresponding maximal Lie invariance algebras.
These subalgebras were used for finding ansatzes for~$u$ and reduced ordinary differential equations.
At the same time, the consideration was needlessly overcomplicated since
the cases of Lie symmetry extensions were not simplified by point equivalence transformations,
and two cases are equivalent to others with respect to point transformations.
Some of the optimal systems are incorrect, cf.~\cite[footnote~7]{poch2014a}.
Moreover, no reduced equations were integrated. 
In \cite{wafo2004d}, Lie reductions were performed only with respect to the one-dimensional subalgebras spanned by single basis elements,
not to mention the presence of equivalent cases and needless parameters in the classification list.
The reduced equation~(95) in \cite{wafo2004d} contains two misprints and should in fact read as $F_0f''+ff'+mz_2\lambda f'-mz_2f+z_2^2\lambda=0$,
cf.~\cite[footnote~8]{poch2014a}.
The further integration procedure is not applicable to the correct version of the reduced equation,
and the functions (99)--(101) in \cite{wafo2004d} do not satisfy
the corresponding generalized Burgers equation.
The only nontrivial solutions (91)--(93) presented in \cite{wafo2004d} look, up to an equivalence transformation,
like particular cases of the solution~\eqref{eqsolutiongeneration3} for the first value
of~$\tilde\omega$ in~\eqref{eqsolutiongeneration5}, $\tilde\omega=\omega/\nu$.
}
and were later enhanced~in~\cite{poch2014a}.
Up~to $G^\sim$-equivalence, Lie symmetry extensions in this subclass are exhausted
by the algebras~$\mathfrak g^{2.1}$, $\mathfrak g^{3.1}$, $\mathfrak g^{3.2}$, $\mathfrak g^{3.3}$ and~$\mathfrak g^5$
of Table~\ref{TableSubalgebras}.

Reduced equations for all possible $G^\sim$-inequivalent one-dimensional subalgebras of~$\mathfrak g$
are presented in Table~\ref{TableRedEqs}.
The case with~$\mathfrak g^{1.0}$ in this table corresponds to
the case with~$\mathfrak g^{2.1}$ in Table~\ref{TableSubalgebras}; cf.\ Lemma~\ref{LemmaOnPxG}.

\begin{table}[ht!]\footnotesize
\begin{center}
\caption{\footnotesize Lie reductions with respect to one-dimensional subalgebras of~$\mathfrak g$
\strut}\label{TableRedEqs}
\renewcommand{\arraystretch}{1.4}
\begin{tabular}{|l|l|l|l|l|}
\hline
\hfil$\mathfrak \subset\mathfrak g$
                     &\hfil Basis &\hfil Ansatz, $\varphi=\varphi(\omega)$
                                                          & \hfil~$\omega$ &\hfil Reduced equation
\\
\hline
$\mathfrak g^{1.0}$  &$P^x$       &$u=\varphi$            &$t$             &$\varphi_\omega=0$
\\[1.5ex]
$\mathfrak g^{1.1}$  &$D^x$       &$u=x\varphi$           &$t$             &$\varphi_\omega+\varphi^2=0$
\\[1.5ex]
$\mathfrak g^{1.2}$  &$P^t$       &$u=\varphi$            &$x$             &$h(\omega)\varphi_{\omega\omega}+\varphi\varphi_\omega=0$
\\[1.5ex]
$\mathfrak g^{1.3}$  &$P^t+G$     &$u=\varphi+t$          &$x-\dfrac{t^2}2$&$h(\omega)\varphi_{\omega\omega}+\varphi\varphi_\omega+1=0$
\\[1.5ex]
$\mathfrak g^{1.4}$  &$P^t+D^x$   &$u=e^t\varphi$         &$e^{-t}x$
&$h(\omega)\varphi_{\omega\omega}+\varphi\varphi_\omega-\omega\varphi_\omega+\varphi=0$
\\[1.5ex]
$\mathfrak g^{1.5}$  &$D^t+P^x$   &$u=t^{-1}\varphi$      &$x-\ln|t|$
&$h(\omega)\varphi_{\omega\omega}+\varphi\varphi_\omega-\varphi_\omega-\varphi=0$
\\[1.5ex]
$\mathfrak g^{1.6}_a$&$D^t+aD^x$  &$u=|t|^at^{-1}\varphi$ &$|t|^{-a}x$
&$h(\omega)\varphi_{\omega\omega}+\varphi\varphi_\omega-\omega\varphi_\omega+(a-1)\varphi=0$
\\[1.5ex]
$\mathfrak g^{1.7}$  &$D^t+D^x+G$ &$u=\varphi+\ln|t|$     &$\dfrac xt-\ln|t|$
&$h(\omega)\varphi_{\omega\omega}+\varphi\varphi_\omega-(\omega+1)\varphi_\omega+1=0$
\\[1.5ex]
$\mathfrak g^{1.8}_a$&$P^t+\Pi+aD^x$
&$u=\dfrac{e^{a\arctan t}}{\sqrt{t^2+1}}\varphi+\dfrac{t+a}{t^2+1}x$ &$\dfrac{e^{-a\arctan t}}{\sqrt{t^2+1}}x\!$
&$h(\omega)\varphi_{\omega\omega}+\varphi\varphi_\omega+2a\varphi+(a^2+1)\omega=0\!$
\\[1.8ex]
\hline
\end{tabular}
\end{center}
\end{table}

In order to solve the second-order reduced equations listed in Table~\ref{TableRedEqs}, we consider the superclass
of ordinary differential equations of the form
\begin{gather}\label{RedEqClass}
h(\omega)\phi_{\omega\omega}+\phi\phi_{\omega}+\alpha\phi+\beta\omega+\gamma=0
\qquad
\text{with}
\qquad
h(\omega)\ne0,
\end{gather}
which contains all the reduced equations (except
those for~$\mathfrak g^{1.0}$ and~$\mathfrak g^{1.1}$).
The change of the variable~$\varphi$ (if needed)
and the values of the constants~$\alpha$, $\beta$ and~$\gamma$ for them are as follows:
\begin{gather*}
\arraycolsep=7pt\renewcommand{\arraystretch}{1.3}
\begin{array}{rllll}
\mathfrak g^{1.2}: & \alpha=0, & \beta=0, & \gamma=0, & \varphi=\phi;
\\
\mathfrak g^{1.3}: & \alpha=0, & \beta=0, & \gamma=1, & \varphi=\phi;
\\
\mathfrak g^{1.4}: & \alpha=2, & \beta=1, & \gamma=0, &\text{after the change}\quad \varphi=\phi+\omega;
\\
\mathfrak g^{1.5}: & \alpha=-1, & \beta=0, & \gamma=-1, &\text{after the change}\quad \varphi=\phi+1;
\\
\mathfrak g^{1.6}_a: & \alpha=a, & \beta=a-1, & \gamma=0, &\text{after the change}\quad \varphi=\phi+\omega;
\\
\mathfrak g^{1.7}: & \alpha=1, & \beta=0, & \gamma=1, &\text{after the change}\quad \varphi=\phi+\omega+1;
\\
\mathfrak g^{1.8}_a: & \alpha=2a, & \beta=a^2+1, & \gamma=0, & \varphi=\phi.
\end{array}
\end{gather*}

Linear solutions of reduced equations of the form~\eqref{RedEqClass},
as well as all solutions of the reduced equations
for~$\mathfrak g^{1.0}$ and~$\mathfrak g^{1.1}$,
lead to solutions of equations from the class~\eqref{GBE} that are linear with respect to~$x$.
The solutions being linear with respect to~$x$ are only common for all equations from the class~\eqref{GBE}
and are exhausted by the two families, $u=c_0$ and $u=(x+c_1)/(t+c_2)$,
where $c_0$, $c_1$ and $c_2$ are arbitrary constants.
They also arise in Section~\ref{sectionTrivialcase} within the framework of reduction operators.

We find Lie symmetries of ordinary differential equations
from the class~\eqref{RedEqClass} and use them for solving reduced equations
presented for the subalgebras \mbox{$\mathfrak g^{1.2}$--$\mathfrak g^{1.8}_a$} from Table~\ref{TableRedEqs}.
For these symmetries to be well interpreted as symmetries of reduced equations,
equivalence transformations between equations from the class~\eqref{RedEqClass}
are not involved in the consideration.

\begin{proposition}\label{RedEqClassification}
The values of the arbitrary elements~$h$, $\alpha$, $\beta$ and $\gamma$
that correspond to equations from the class~\eqref{RedEqClass}
with nonzero maximal Lie invariance algebras, $\mathfrak h$, are exhausted by
\begin{enumerate}
\item $h=h_0\big(\omega+\frac{\gamma}{\beta}\big)^2$, $\beta\ne0\colon$
 $\mathfrak h=\langle\big(\omega+\frac{\gamma}{\beta}\big)\partial_{\omega}+\phi\partial_{\phi}\rangle$;
\item $h=h_0$, $\beta=0$, $\gamma\ne0\colon$
 $\mathfrak h=\langle\partial_{\omega}\rangle$;
\item $h=h_0|\omega+\mu|^{3/2}$, $\alpha=\beta=0$, $\gamma\ne0\colon$
 $\mathfrak h=\langle2(\omega+\mu)\partial_{\omega}+\phi\partial_{\phi}\rangle$;
\item $h=-\frac{\alpha}{2}\omega^2+\mu\omega+\nu$, $\beta=\gamma=0\colon$
 $\mathfrak h=\langle h\partial_{\omega}-\alpha h\partial_{\phi}$,
 $-\big(h\!\!\int\!\frac{\d\omega}{h}\big)\partial_{\omega}+\big(\phi+\alpha h\!\!\int\!\frac{\d\omega}{h}+\alpha\omega-\mu\big)\partial_{\phi}\rangle$;
\item $h_{\omega\omega}=\frac{\kappa}{h}-\alpha$, $\beta=\gamma=0\colon$
 $\mathfrak h=\langle h\partial_{\omega}+(\kappa-\alpha h)\partial_{\phi}\rangle$;
\item $\frac{h_{\omega\omega}+\alpha}{(h_{\omega}+\alpha\omega+\mu)^2}=\frac{\kappa}{h}$, $\beta=\gamma=0\colon$
 $\mathfrak h=\langle \xi\partial_\omega+(\phi-\alpha\xi+\alpha\omega+\mu)\partial_\phi\rangle$.
\end{enumerate}
Here $h_0$, $\mu$, $\nu$ and $\kappa$ are arbitrary constants with $h_0\ne0$ and $\kappa\ne0$,
and $\xi=\frac{h_{\omega}+\alpha\omega+\mu}{h_{\omega\omega}+\alpha}$.
\end{proposition}

\begin{proof}
Any Lie symmetry operator of an equation from the class~\eqref{RedEqClass} has
the general form $\xi(\omega)\partial_\omega +[c_1\phi-\alpha\xi(\omega)+\alpha c_1\omega+c_0]\partial_\phi$,
where $c_1$ and $c_0$ are arbitrary constants
and the component $\xi=\xi(\omega)$ satisfies the
classifying equations
\begin{gather}
(\xi_{\omega}-2c_1)(\beta\omega+\gamma)+\beta\xi=0,
\label{RedEqClassDE8}
\\
\xi_\omega-\frac{h_\omega}{h}\xi=-c_1,
\label{RedEqClassDE2}
\\
h\xi_{\omega\omega}=-\alpha\xi+c_1\alpha\omega+c_0,
\qquad
\text{and~hence}
\qquad
(h_{\omega\omega}+\alpha)\xi=c_1h_\omega+c_1\alpha\omega+c_0.
\label{RedEqClassDE7}
\end{gather}
Both the forms of the equation~\eqref{RedEqClassDE7} are equivalent when~\eqref{RedEqClassDE2} holds,
and are useful for the further classification.
The simplest way to solve the system \eqref{RedEqClassDE8}--\eqref{RedEqClassDE7}
is to start its integration from equation~\eqref{RedEqClassDE8}
considering the cases $\beta=0$ and $\beta\ne0$ separately.

For $\beta\ne0$, the equation~\eqref{RedEqClassDE8}
implies
\begin{gather*}
\xi=\frac{c_1(\beta\omega+\gamma)}{\beta}+\frac{b}{\beta\omega+\gamma}.
\end{gather*}
Here and below~$b$ denotes an integration constant.
The assumption $b\ne0$ leads to a~contradiction.
As a~result, $b=0$, i.e.\ $\xi=c_1(\omega+\frac{\gamma}{\beta})$, and thus $\xi_{\omega\omega}=0$.
Hence from the first form of~\eqref{RedEqClassDE7} we get $c_0=c_1\frac{\alpha\gamma}{\beta}$.
Since the algebra~$\mathfrak h$ is supposed to be nonzero, it should contain a~vector field with a~nonzero value of~$c_1$.
Therefore, from~\eqref{RedEqClassDE2} we have $(\omega+\frac{\gamma}{\beta})h_{\omega}=2h$,
which gives item~1 of the proposition.

If $\beta=0$ and $\gamma\ne0$, then from~\eqref{RedEqClassDE8} we obtain $\xi=2c_1\omega+b$
and split the first form of~\eqref{RedEqClassDE7} with respect to~$\omega$ to derive $\alpha c_1=0$ and $c_0=\alpha b$.
Then we exploit~\eqref{RedEqClassDE2} and obtain items~2 and~3 
depending on whether $c_1$ vanishes for all elements of~$\mathfrak h$ or not, respectively.

The equation~\eqref{RedEqClassDE8} with $\beta=\gamma=0$ is an identity.
Consider the subcases $h_{\omega\omega}+\alpha=0$ and $h_{\omega\omega}+\alpha\ne0$ separately.

If $h_{\omega\omega}+\alpha=0$, integrating~\eqref{RedEqClassDE2}
we derive $\xi=bh-c_1h\int\frac{\d\omega}{h}$.
This gives item~4.

For $h_{\omega\omega}+\alpha\ne0$, the second form of the equation~\eqref{RedEqClassDE7} gives
\begin{gather*}
\xi=\frac{h_{\omega}+\alpha\omega}{h_{\omega\omega}+\alpha}c_1+\frac1{h_{\omega\omega}+\alpha}c_0.
\end{gather*}
Then the equation~\eqref{RedEqClassDE2} can be represented as $K^1c_1+K^0c_0=0$,
where
\begin{gather*}
K^0=\left[\frac1{h(h_{\omega\omega}+\alpha)}\right]_{\omega}h
\qquad
\text{and}
\qquad
K^1=2+(h_{\omega}+\alpha\omega)K^0.
\end{gather*}
If $K^0=0$, then $K^1=2$, $c_1=0$ and $\kappa:=h(h_{\omega\omega}+\alpha)=\const$.
That is why we have item~5.
Now suppose that $K^0\ne0$.
$K^1$ and $K^0$ are linearly dependent, otherwise $c_1=c_0=0$,
which corresponds to the trivial algebra.
Therefore, $\mu:=-K^1/K^0=\const$, and hence $c_0=\mu c_1$.
Then the equation~\eqref{RedEqClassDE2} reduces to
\begin{gather*}
\frac{h}{h_{\omega}+\alpha\omega+\mu}\left[\frac{(h_{\omega}+\alpha\omega+\mu)^2}{h(h_{\omega\omega}+\alpha)}\right]_{\omega}=0,
\qquad
\text{or, equivalently,}
\qquad
\frac{h_{\omega\omega}+\alpha}{(h_{\omega}+\alpha\omega+\mu)^2}=\frac{\kappa}{h}
\end{gather*}
when once integrated.
Here $\kappa$ is an integration constant.
Thus we obtain item~6.
In the last two items we have $\kappa\ne0$ since $h_{\omega\omega}+\alpha\ne0$.
\end{proof}

The equation~\eqref{RedEqClass} with $h=-\frac{\alpha}{2}\omega^2+\mu\omega+\nu$
and $\beta=\gamma=0$
(item~4 of Proposition~\ref{RedEqClassification})
admits the widest (two-dimensional) symmetry algebra.
In this case in the variables
\begin{gather*}
\tilde\omega=\int\frac{\d\omega}{-\tfrac{\alpha}{2}\omega^2+\mu\omega+\nu},
\qquad
\tilde\phi=\phi+\alpha\omega-\mu
\end{gather*}
the equation~\eqref{RedEqClass} can
be once integrated, which leads to
$2\tilde\phi_{\tilde\omega}=c_0-\tilde\phi^2$,
where $c_0$ is an integration constant.
The solution of the integrated equation
depends on the sign of $c_0$,
\begin{gather}
\tilde\phi=-\varkappa\tan\left(\frac{\varkappa}{2}\tilde\omega+c_1\right) \qquad\text{if}\quad c_0<0,\ \varkappa:=\sqrt{-c_0},
\nonumber
\\
\tilde\phi=\frac2{\tilde\omega+c_1}
\quad\text{or}\quad
\tilde\phi=0 \qquad\text{if}\quad c_0=0,
\label{eqsolutiongeneration3}
\\
\tilde\phi=\varkappa\dfrac{c_1e^{\varkappa\tilde\omega}-1}{c_1e^{\varkappa\tilde\omega}+1}
\quad\text{or}\quad
\tilde\phi=\varkappa \qquad\text{if}\quad c_0>0,\ \varkappa:=\sqrt{c_0},
\nonumber
\end{gather}
where $c_1$ is another integration constant.
The form of $\tilde\omega$ depends on the sign of $\Delta:=\mu^2+2\alpha\nu$ and on $\alpha$, $\mu$ and $\nu$,
namely
\begin{gather}
\tilde\omega=
\begin{cases}
\dfrac{\omega}{\nu}, & \alpha=0, \; \mu=0, \; \nu\ne0,
\\[3mm]
\dfrac{1}{\mu}\ln\left|\omega+\dfrac{\nu}{\mu}\right|, & \alpha=0, \; \mu\ne0,
\\[3.9mm]
-\dfrac{2}{\sqrt{-\Delta}}\arctan\dfrac{\alpha\omega-\mu}{\sqrt{-\Delta}}, & \alpha\ne0, \; \Delta<0,
\\[3.9mm]
\dfrac{2}{\alpha\omega-\mu}, & \alpha\ne0, \; \Delta=0,
\\[3mm]
\dfrac{1}{\sqrt\Delta}\ln\left|\dfrac{\alpha\omega-\mu+\sqrt{\Delta}}{\alpha\omega-\mu-\sqrt\Delta}\right|, & \alpha\ne0, \; \Delta>0.
\end{cases}
\label{eqsolutiongeneration5}
\end{gather}
Therefore, substituting the expressions for $\tilde\phi$ and $\tilde\omega$ into
$\phi=\tilde\phi(\tilde\omega)-\alpha\omega+\mu$
we have fifteen different expressions for solutions of~\eqref{RedEqClass}
with~$h$ quadratic in $\omega$ and $\beta=\gamma=0$.

We can apply this result to the reduced equations that correspond
to the subalgebras~$\mathfrak g^{1.2}$ and~$\mathfrak g^{1.6}_1$ (cf.~Table~\ref{TableRedEqs}),
where $\alpha=0$ and $\alpha=1$, respectively.

The differential equation for~$h$ in item~5 being multiplied by $2h_{\omega}$ and once integrated
takes the form $h_{\omega}^2=2\nu\ln|h|-2\alpha h+c_1$, where $c_1$ is an arbitrary constant.
After the second integration we have the implicit general solution
\begin{gather*}
\pm\int\frac{\d h}{\sqrt{2\nu\ln|h|-2\alpha h+c_1}}=\omega+c_2.
\end{gather*}

For some items of Proposition~\ref{RedEqClassification},
we can construct at least particular solutions of related equations of the form~\eqref{RedEqClass}
for certain values of parameters.
We can use e.g.\ Lie reduction of equations from the class~\eqref{RedEqClass} to algebraic equations.

The symmetry algebra of item~3 gives the ansatz
$\phi=c\sqrt{|\omega+\mu|}$.
The constant $c$ is a~solution of the reduced algebraic equation
$2\varepsilon c^2-h_0c+4\gamma=0$, where $\varepsilon=\sign(\omega+\mu)$, and hence
\begin{gather*}
c=\frac1{4\varepsilon}\big(h_0\pm\sqrt{h_0^2-32\varepsilon\gamma}\,\big).
\end{gather*}
Hereby for $\gamma=1$ we can use this solution $\phi$ of the reduced equation associated with the subalgebra~$\mathfrak g^{1.3}$,
see Table~\ref{TableRedEqs}, to find the solution
\begin{gather}
\label{sqrtSolution}
u(t,x)=c\sqrt{\left|x-\frac{t^2}{2}+\mu\right|}+t
\end{gather}
of the generalized Burgers equation~$\mathcal L_f$ with $f=h_0\big|x-\frac{t^2}{2}+\mu\big|^{3/2}$,
and $\mu=0\bmod G^\sim$.

If we suppose $\alpha=\mu=0$ in item~6, then the corresponding equation $h_{\omega\omega}/h_{\omega}=\kappa h_{\omega}/h$
for~$h$ can be integrated,
$h=h_0|\omega+\lambda|^{\frac{1}{1-\kappa}}$ if $\kappa\ne1$ and $h=h_0e^{\lambda\omega}$ if $\kappa=1$,
\noprint{
\begin{gather*}
h=h_0|\omega+\lambda|^{\frac{1}{1-\kappa}} \qquad\text{if}\quad \kappa\ne1,
\\
h=h_0e^{\lambda\omega} \qquad\text{if}\quad \kappa=1,
\end{gather*}
}
and the associated Lie symmetry algebra leads to the ansatz $\phi=c(h/h_0)^{\kappa}$.
Here $h_0$ and $\lambda$ are integration constants.
This ansatz reduces
the equation~\eqref{RedEqClass} to a~quadratic
equation in $c$ whose two solutions are $c=0$~and
\begin{gather*}
c=
\begin{cases}
\dfrac{1-2\kappa}{1-\kappa}h_0\sign(\omega+\lambda) &\quad\text{if}\quad \kappa\ne1,
\\
-\lambda h_0 &\quad\text{if}\quad \kappa=1.
\end{cases}
\end{gather*}
Using this result for the reduced equation that corresponds to the algebra~$\mathfrak g^{1.2}$,
see Table~\ref{TableRedEqs}, we obtain the stationary solutions
$u(t,x)=c|x+\lambda|^{\frac{\kappa}{1-\kappa}}$
and
$u(t,x)=-\lambda h_0e^{\lambda x}$
for the equations of the form~\eqref{GBE}
with $f=h_0|x+\lambda|^{\frac{1}{1-\kappa}}$ and $f=h_0e^{\lambda x}$,
respectively.

\begin{remark}
If a~reduced equation admits Lie symmetries
that are not induced by Lie symmetries of the initial equation,
then the initial equation is said to have \emph{additional}~\cite{olve1993b}
(or \emph{hidden}~\cite{abra2008a}) symmetries
with respect to the corresponding reduction.
The first example of such symmetries was presented in~\cite{kapi1978a};
see also the discussion of this example in~\cite[Example~3.5]{olve1993b}.
A~comprehensive study of such symmetries for the Navier--Stokes equations
was carried out in~\cite{fush1994a,fush1994b}.
The Lie reductions of relevant equations from the class~\eqref{GBE}
with respect to algebras~$\mathfrak g^{1.0}$ and~$\mathfrak g^{1.1}$
lead to first-order reduced equations.
Therefore, the corresponding initial equations admit infinite-dimensional families of hidden symmetries
with respect to the above reductions,
but these symmetries are not essential for consideration because they provide no new exact solutions.
The other algebras from Table~\ref{TableRedEqs} give second-order reduced equations
of the general form~\eqref{RedEqClass}.
Among Lie symmetries of such reduced equations there are
both induced and hidden symmetries.
Namely,
in items 1--3 of Proposition~\ref{RedEqClassification},
all symmetries of related reduced equations
(which are constructed using the algebras
$\mathfrak g^{1.4}$, $\mathfrak g^{1.6}_a$ with $a\ne1$, $\mathfrak g^{1.8}_a$;
$\mathfrak g^{1.3}$, $\mathfrak g^{1.5}$, $\mathfrak g^{1.7}$;
$\mathfrak g^{1.3}$, respectively) are induced
by Lie symmetries of the corresponding initial equations from the class~\eqref{GBE}.
The condition $\beta=\gamma=0$ of items~4--6 can be satisfied
only by reduced equations obtained using the algebras~$\mathfrak g^{1.2}$ and~$\mathfrak g^{1.6}_1$.
In item~4, the first basis vector field is induced only if $f=x$
for the reduction with respect to~$\mathfrak g^{1.2}$,
and the second basis vector field is induced only if
$f=t(\omega+\bar\nu)^2$ with $\omega=x/t$ for the reduction with respect to~$\mathfrak g^{1.6}_1$.
All the other Lie symmetries of reduced equations presented in items~4--6
are hidden symmetries of the corresponding initial equations from the class~\eqref{GBE}.
\end{remark}

Consider possible reductions with respect to the two-dimensional
inequivalent subalgebras of the algebra~$\mathfrak g$.
For this purpose for each basis vector field of a~subalgebra we write the characteristic equation
and thus obtain a~system of two differential equations.
If the system is consistent, then its solution gives an ansatz
reducing the corresponding equations from the class~\eqref{GBE} to algebraic equations.
The system obtained for~$\mathfrak g^{2.1}$ is inconsistent,
i.e.\ there is no ansatz associated with this subalgebra.
An ansatz constructed with~$\mathfrak g^{2.4}$ is $u=(t+c)x/(t^2+1)$
but the corresponding reduced equation $c^2+1=0$ has no solutions.
The other subalgebras allow us to construct some simple solutions, namely,
$u=0$ from~$\mathfrak g^{2.2}$,
$u=0$ and $u=x/t$ from~$\mathfrak g^{2.3}$,
$u=0$ and $u=e^{-x}$ from~$\mathfrak g^{2.5}$,
$u=0$ and $u=a^{-1}x|x|^{-1/a}$ from~$\mathfrak g^{2.6}_a$,
and from~$\mathfrak g^{2.7}$ we re-obtain solution~\eqref{sqrtSolution} with $\mu=0$ and $h_0=\kappa$.
Up to $G^\sim$-equivalence, 
the other two-dimensional subalgebras of~$\mathfrak g$ can be assumed to contain 
the vector field~$P^x$~\cite{gaze1992a}
and, therefore, to lead at most to constant solutions 
of equations from the class~\eqref{GBE}.

\section{Reduction operators and nonclassical reductions}\label{ROsection}

Achieving new possible reductions opens a~way for finding more exact solutions,
which may be of interest for modeling physical phenomena
and verifying approximate methods of solving differential equations.
Nonclassical reductions were first considered in~\cite{blum1969a} as a~generalization of the classical Lie reduction method.
An attempt of formalizing them was made in~\cite{fush1987a}.
Vector fields associated with nonclassical reductions are called nonclassical, or conditional, or $Q$-conditional symmetries~\cite{fush1993d,olve1996a}.
Another, more proper, name for such a~vector field~$Q$ is a~\emph{reduction operator}~\cite{kunz2009a},
which relates~$Q$ to reducing the number of independent variables
in a~partial differential equation with an ansatz constructed by~$Q$~\cite{zhda1999a}.

After the linear heat equation~\cite{blum1969a},
the Burgers equation was the second one that was considered
from the point of view of reduction operators~\cite{wood1971b,wood1971a};
see also a~review of these results in~\cite{ames1972a}.
Later, reduction operators of the Burgers equation were objects of study and discussion
in a~number of papers~\cite{arri1993a,arri2002a,mans1999a,olve1996a,pucc1992a}.
These studies were summed up in~\cite{poch2013a}.
Attempts to describe nonclassical symmetries for equations from the class~\eqref{GBE}
with nonconstant~$f$'s were started in~\cite{wafo2004d} for the subclass of equations with $f_x=0$.
It was also shown that in this subclass
reduction operators inequivalent to Lie symmetries exist only for
$f=\const$, i.e.\ for the classical Burgers equation.
Preliminary results on reduction operators of equations from the class~\eqref{GBE} with general~$f$'s
were first outlined in~\cite{poch2012a}.

Here we arrange the consideration of reduction operators for generalized
Burgers equations presented in~\cite{poch2012a,poch2013a}
and extend it with complete proofs of the corresponding assertions.

Roughly speaking,
a~\emph{reduction operator} of an equation~$\mathcal L_f$ from the class~\eqref{GBE} is a~vector field of the form
$Q=\tau(t,x,u)\p_t+\xi(t,x,u)\p_x+\eta(t,x,u)\p_u$ with $(\tau,\xi)\ne(0,0)$
that leads to an ansatz reducing the initial equation to
an ordinary differential equation (see~\cite{kunz2008b,popo2008b,zhda1999a} for precise definitions).
Due to the equivalence relation of reduction operators, one can multiply~$Q$
by a~nonzero function of~$(t,x,u)$ in order to gauge a~component of~$Q$ to one.
The set of reduction operators for any (1+1)-dimensional evolution equation
can be naturally partitioned into two subsets depending on whether $\tau=0$ or $\tau\ne0$.
Moreover, reduction operators with $\tau=0$ are \emph{singular}~\cite{boyk2016a,kunz2008b},
and the problem of finding them is equivalent to solving a single determining equation, 
which reduces to the original equation~\cite{fush1992e,kunz2008b,zhda1998a}
(so-called ``no-go'' problem).%
\footnote{%
Sometimes such a~determining equation may be useful, when
it is possible to guess some ad hoc forms of its particular solutions,
although there is no algorithmic procedure to do this.
See, e.g., \cite{gand2001b}.
Singular reduction operators corresponding to these particular solutions
can be used to construct exact solutions for the original equation.\looseness=-1
}
In particular, the determining equation for singular reduction operators of the equation~$\mathcal L_f$ is
\begin{gather*}
\eta_t+u\eta_x+\eta^2+f_x(\eta_x+\eta\eta_u)+f(\eta_{xx}+2\eta\eta_{xu}+\eta^2\eta_{uu})=0,
\end{gather*}
where the component~$\xi$ is already set to 1
using the equivalence relation of reduction operators;
cf.~\cite[Section~2]{poch2013a} for the Burger equation~$\mathcal L_1$.

For this reason, we devote the rest of this section to \emph{regular} reduction operators
of equations from the class~\eqref{GBE},
which have, up to the equivalence relation of reduction operators,
the general form $Q=\p_t+\xi(t,x,u)\p_x+\eta(t,x,u)\p_u$.
The conditional invariance criterion~\cite{fush1993d,popo2007a,zhda1999a} implies the condition
\begin{equation}\label{eq:CondInvCriterion}
Q_{(2)}L^f[u]\,\big|_{\mathcal L_f\cap\mathcal Q_{(2)}}=0
\end{equation}
for the vector field~$Q$ to be a~reduction operator of the equation~$\mathcal L_f$: $L^f[u]=0$.
Here, again, $Q_{(2)}$ is the standard second-order prolongation of the vector field~$Q$,
the manifold defined by the equation~$\mathcal L_f$
in the second-order jet space with the variables $(t,x,u,u_t,u_x,u_{tt},u_{tx},u_{xx})$
is denoted by the same symbol~$\mathcal L_f$,
and~$\mathcal Q_{(2)}$ is the manifold defined in the same jet space by
the invariant surface condition $Q[u]:=\eta-u_t-\xi u_x=0$
and its differential consequences ${\rm D}_tQ[u]=0$ and ${\rm D}_xQ[u]=0$.
These consequences are not needed
in the course of expanding the condition~\eqref{eq:CondInvCriterion}
since the expression $Q_{(2)}L^f[u]$ does not contain the derivatives~$u_{tt}$ and~$u_{tx}$
due to the gauging of the component~$\tau$ of~$Q$ to~1.
Thus, the expanded condition~\eqref{eq:CondInvCriterion}~is
\begin{gather*}
\eta^t+\eta u_x+u\eta^x+(f_t+\xi f_x)u_{xx}+f\eta^{xx}=0
\qquad
\text{if}
\qquad
u_t+uu_x+fu_{xx}=0,
\
u_t+\xi u_x=\eta.
\end{gather*}
Substituting $u_t=\eta-\xi u_x$ and $u_{xx}=(\xi u_x-uu_x-\eta)/f$
and splitting the result with respect to~$u_x$, we obtain the system of determining equations
\begin{gather}
\xi_{uu}=0,
\quad
\eta_{uu}=\frac2f\xi_u(\xi-u)+2\xi_{xu},
\nonumber
\\
(2\xi_u+1)\eta+\left(\frac{f_t}{f}+\frac{f_x}{f}\xi\right)(\xi-u)+2f\eta_{xu}-\xi_t-2\xi_x\xi+u\xi_x-f\xi_{xx}=0,
\nonumber
\\
\eta_t+u\eta_x+f\eta_{xx}-\left(\frac{f_t}{f}+\frac{f_x}{f}\xi\right)\eta+2\xi_x\eta=0.\label{de7}
\end{gather}
Integrating the first two equations,
we can represent $\xi$ and $\eta$ as
polynomials of~$u$ with coefficients depending on~$t$ and $x$,
\begin{gather*}
\xi=\xi^1(t,x)u+\xi^0(t,x),
\quad
\eta=\frac{\xi^1\left(\xi^1-1\right)}{3f}u^3
+\left(\xi^1_x+\frac{\xi^1\xi^0}{f}\right)u^2+\eta^1(t,x)u+\eta^0(t,x),
\end{gather*}
where the coefficients $\xi^1$, $\xi^0$, $\eta^1$ and $\eta^0$ are assumed as new unknown functions.
Substituting the expressions for~$\xi$ and~$\eta$ into the third determining equation
and splitting the result with respect to~$u$, we derive the system
\begin{gather}\label{GBEROde}
\begin{split}
&\xi^1(\xi^1-1)(2\xi^1+1)=0,
\\[.5ex]
&\xi^1(2\xi^1+1)\xi^0-\xi^1(\xi^1-1)f_x  
=0,\\[.5ex]
&(\xi^1-1)f_t+(2\xi^1+1)(f\eta^1+f\xi^0_x-f_x\xi^0) 
=0,\\
&(2\xi^1+1)\eta^0+\left(\frac{f_t}{f}+\frac{f_x}{f}\xi^0\right)\xi^0+2f\eta^1_x
-\xi^0_t-2\xi^0\xi^0_x-f\xi^0_{xx}=0.
\end{split}
\end{gather}
In the second and third equations of~\eqref{GBEROde}
we immediately take into account the implication $\xi^1=\const$ of the first equation of~\eqref{GBEROde}.

The further consideration depends on the choice of solution of the first equation of~\eqref{GBEROde}.
We devote the next subsections to the cases $\xi^1=1$ and $\xi^1=-\frac12$,
and the case $\xi^1=0$ is partitioned into two subcases, $\xi^0_{xx}=0$ and $\xi^0_{xx}\ne0$.
Note that the last determining equation~\eqref{de7}
will be represented in terms of $\xi^1$, $\xi^0$, $\eta^1$ and $\eta^0$ and split
with respect to~$u$ in every particular case.

The partition of reduction operators of equations from the class~\eqref{GBE}
into the singular and regular reduction operators
and the further partition of the regular case into the above subcases
are invariant under the action of~$G^\sim$ on the pairs (`equation', `its reduction operator').
See, e.g., \cite[Section~3]{popo2008b} or \cite[Definition~3]{kunz2008b}.
Therefore, these reduction operators can be classified up to $G^\sim$-equivalence,
which coincides with $\mathcal G^\sim$-equivalence
due to the normalization of the class~\eqref{GBE}.

The following theorem holds.

\begin{theorem}\label{Theorem_RO_subsets}
Up to $G^\sim$-equivalence,
all regular reduction operators
of equations from the class~\eqref{GBE} are exhausted~by
\begin{enumerate}\setSkips
 \item
$Q^1=\partial_t+u\partial_x$ for any equation~$\mathcal L_f$ from the class~\eqref{GBE},
 \item
Lie symmetry operators with nonzero coefficients of $\partial_t$,
 \item

$Q^\theta=\partial_t-(\theta_t/\theta_x)\partial_x$
for each equation~$\mathcal L_{f^\theta}$ with $f^\theta=-1/\theta_x$,
where $\theta=\theta(t,x)$ is an~arbitrary nonconstant solution of the equation
\begin{equation*}
\theta_t=\frac{\theta_{xx}}{\theta_x}+h(\theta)\theta_x,
\end{equation*}
and~$h$ is an arbitrary smooth function of~$\theta$,
 \item
$Q^{\xi^0\eta^1\eta^0}\!=\partial_t+\big(\!-\frac1{2}u+\xi^0\big)\partial_x
+\big(\frac1{4}u^3-\frac1{2}\xi^0u^2+\eta^1u+\eta^0\big)\partial_u$
only for the classical Burgers equation~$\mathcal L_1$
(modulo $G^\sim$-equivalence, any constant $f$ can be set to one),
where
\begin{gather*}
\xi^0=\frac1{2}
\frac{\left|\begin{matrix}1&u^1&z^1\\1&u^2&z^2\\1&u^3&z^3\end{matrix}\right|}
{\left|\begin{matrix}1&u^1&y^1\\1&u^2&y^2\\1&u^3&y^3\end{matrix}\right|}\,,
\qquad
\eta^1=\frac1{4}
\frac{\left|\begin{matrix}1&y^1&z^1\\1&y^2&z^2\\1&y^3&z^3\end{matrix}\right|}
{\left|\begin{matrix}1&u^1&y^1\\1&u^2&y^2\\1&u^3&y^3\end{matrix}\right|}\,,
\qquad
\eta^0=-\frac1{4}
\frac{\left|\begin{matrix}u^1&y^1&z^1\\u^2&y^2&z^2\\u^3&y^3&z^3\end{matrix}\right|}
{\left|\begin{matrix}\,1\,&u^1&y^1\\\,1\,&u^2&y^2\\\,1\,&u^3&y^3\end{matrix}\right|}\,,
\end{gather*}
$u^i$ are solutions of~$\mathcal L_1$ such that the determinant in the denominators does not vanish,
$y^i=2u^i_x+(u^i)^2$ and $z^i=4u^i_{xx}+6u^iu^i_x+(u^i)^3$, $i=1,2,3$.
\end{enumerate}
\end{theorem}

The proof of the theorem is split into the four parts corresponding to theorem's cases and accompanied with brief discussions of associated nonclassical reductions.

\subsection{\texorpdfstring{Trivial case $\boldsymbol{\xi^1=1}$}{Trivial case xi^1=1}}\label{sectionTrivialcase}

The case $\xi^1=1$ was considered in~\cite{arri1993a,pucc1992a} for the classical Burgers equation
and in~\cite{poch2012a} for all equations from the class~\eqref{GBE}.
The determining equations~\eqref{GBEROde} imply $\xi^0=0$, $\eta^1=0$, $\eta^0=0$,
which gives \emph{item~1} of Theorem \ref{Theorem_RO_subsets}.
The vector field $Q^1=\partial_t+u\partial_x$
is a~unique common reduction operator for all equations of the class~\eqref{GBE}.
The set of $Q^1$-invariant solutions of every equation~$\mathcal L_f$ from the class~\eqref{GBE}
is exhausted by two families,
$u(t,x)=(x+c_1)/(t+c_2)$ and $u=c$,
where $c_1$, $c_2$ and~$c$ are arbitrary constants,
and these are only common solutions
for all equations from the class~\eqref{GBE}, cf.~Section~\ref{SOLUTIONSsection}.
See \cite[Section~3.1]{poch2013a} for more details on related reductions. 
Note that any transformation from the restriction of~$G^\sim$ to the space $(t,x,u)$ 
pushes forward~$Q^1$ to an equivalent vector field differing from~$Q^1$ by a nonvanishing multiplier 
and preserves both the families of $Q^1$-invariant solutions.

\subsection{\texorpdfstring{Case $\boldsymbol{\xi^1=-\frac12}$}{Case xi1=-1/2}}

This value for $\xi^1$ is possible only if $f=\const$, hence the subproblem in question is equivalent to the
problem of finding reduction operators with $\xi^1=-\frac12$ for the classical Burgers equation.
Since the usual equivalence group contains scale transformations impacting the constant~$f$, we can
set $f=1$ modulo $G^\sim$, and thus obtain the Burgers equation
\[
\mathcal L_1\colon\quad L^1[u]:=u_t+uu_x+u_{xx}=0.
\]

The determining equations on the components
$\xi^0=\xi^0(t,x)$, $\eta^1=\eta^1(t,x)$ and $\eta^0=\eta^0(t,x)$
become
\begin{gather}\label{de_case12_1}
\begin{split}
&\xi^0_t +2\xi^0_x\xi^0 +\xi^0_{xx}-2\eta^1_x=0,
\\
&\eta^1_t+2\xi^0_x\eta^1+\eta^1_{xx}+\eta^0_x=0,
\\
&\eta^0_t+2\xi^0_x\eta^0+\eta^0_{xx}=0.
\end{split}
\end{gather}
Reduction operators in this case take the form
\begin{gather}\label{QforCase12}
Q=\p_t+\left(-\frac1{2}u+\xi^0\right)\p_x
+\left(\frac1{4}u^3-\frac{\xi^0}{2}u^2+\eta^1u+\eta^0\right)\p_u.
\end{gather}

As it was established in~\cite{arri2002a,mans1999a}, solving the system~\eqref{de_case12_1}
is equivalent to solving the system of three copies of the linear heat equation.
Therefore, it has been referred to as a~``no-go'' problem.
Using a~technique developed in~\cite{popo2008b},
we showed in~\cite{poch2013a} that this equivalence immediately follows from the fact
that $Q$ is a~reduction operator of the Burgers equation.
Moreover, we derived the representation of solutions of the system~\eqref{de_case12_1}
via solutions of the uncoupled system of three copies of the initial Burgers equation.
Below we briefly review results of~\cite{poch2013a}.
Note that the proofs given in \cite{arri2002a,mans1999a} did not use
the relation of the system~\eqref{de_case12_1} with reduction operators of the Burgers equation.

\begin{lemma}\label{Theorem_coeffs_via_v}
Any solution of the determining system~\eqref{de_case12_1}
on the coefficients of reduction operators of the form~\eqref{QforCase12} is represented as
\begin{equation}\label{BE_RO_case12_coeffs_via_v}
\xi^0=\frac{(W(\bar v))_x}{W(\bar v)},
\quad
\eta^1=\displaystyle\frac{|\bar v,\bar v_{xx},\bar v_{xxx}|}{W(\bar v)},
\quad
\eta^0=-2\displaystyle\frac{W(\bar v_x)}{W(\bar v)},
\end{equation}
where $\bar v=(v^1,v^2,v^3)$ is a~triple of linearly independent solutions of the heat equation $v_t+v_{xx}=0$,
$W(\bar v)=|\bar v,\bar v_x,\bar v_{xx}|$ and~$W(\bar v_x)=|\bar v_x,\bar v_{xx},\bar v_{xxx}|$
are the Wronskians of this triple
and the triple of the corresponding derivatives with respect to~$x$, respectively,
and $|\bar p,\bar q,\bar r|$ denotes
the determinant of the matrix constructed with ternary columns~$\bar p$, $\bar q$ and~$\bar r$.
Conversely, any triple~$(\xi^0,\eta^1,\eta^0)$ admitting the representation~\eqref{BE_RO_case12_coeffs_via_v}
satisfies the system~\eqref{de_case12_1}.
\end{lemma}

The proof of the lemma is based on properties of reduction operators.
Considering $Q$-invariant solutions for an operator~$Q$ of the form~\eqref{QforCase12},
we solve the system of two equations $L^1[u]=0$ and $Q[u]=0$,
recombined for convenience as $L^1[u]=0$, $L^1[u]+Q[u]=0$.
The Hopf--Cole transformation $u=2v_x/v$ maps this system to the linear system
\begin{gather}
v_t+v_{xx}=0,
\label{BE_case12_heat}
\\
v_{xxx}-\xi^0v_{xx}+\eta^1v_x+\frac12\eta^0v=0.
\label{BE_case12_vxxx}
\end{gather}
Since the family of $Q$-invariant solutions of~$\mathcal L_1$ is two-parameter,
the space of solutions of the system~\eqref{BE_case12_heat}--\eqref{BE_case12_vxxx}
with respect to~$v$ is three-dimensional.
Let the functions $v^i=v^i(t,x)$, $i=1,2,3$, constitute a~basis of this space.
Then the family of $Q$-invariant solutions of~$\mathcal L_1$ is represented as
\begin{gather}\label{EqTwoParamFamilyOfSolutionInTermsOfSolutionsOfLHE}
u=2\frac{c_1v^1_x+c_2v^2_x+c_3v^3_x}{c_1v^1+c_2v^2+c_3v^3},
\end{gather}
where only two of the constants $c_1$, $c_2$, $c_3$ are essential.

Now take three arbitrary linearly independent solutions $v^i$ of the heat equation~\eqref{BE_case12_heat}
and substitute each of them into~\eqref{BE_case12_vxxx}.
Solving the obtained system of three copies of~\eqref{BE_case12_vxxx} as
a~system of linear algebraic equations on~$\xi^0$, $\eta^0$ and~$\eta^1$,
we derive the representation~\eqref{BE_RO_case12_coeffs_via_v}.

Conversely, if the coefficients~$\xi^0$, $\eta^0$ and~$\eta^1$
are of the form~\eqref{BE_RO_case12_coeffs_via_v},
then the equation~$\mathcal L_1$ admits
a~two-parametric family~\eqref{EqTwoParamFamilyOfSolutionInTermsOfSolutionsOfLHE}
of $Q$-invariant solutions.
Hence the vector field~$Q$ is a~reduction operator of~$\mathcal L_1$,
which completes the proof of Lemma~\ref{Theorem_coeffs_via_v}.

Lemma~\ref{Theorem_coeffs_via_v} and the Hopf--Cole transformation
allow us to represent the coefficients~$\xi^0$, $\eta^0$ and~$\eta^1$
in terms of three solutions of the classical Burgers equation,
\begin{gather}
\xi^0=\frac1{2}
\frac{|\bar e,\bar u,\bar z|}
{|\bar e,\bar u,\bar y|},
\quad
\eta^1=\frac1{4}
\frac{|\bar e,\bar y,\bar z|}
{|\bar e,\bar u,\bar y|},
\quad
\eta^0=-\frac1{4}
\frac{|\bar u,\bar y,\bar z|}
{|\bar e,\bar u,\bar y|},
\end{gather}
where the columns~$\bar e$, $\bar y$ and~$\bar z$ consist of
$e^i=1$, $y^i=2u^i_x+(u^i)^2$ and~$z^i=4u^i_{xx}+6u^iu^i_x+(u^i)^3$, respectively, $i=1,2,3$,
and $\bar u$ is a~column of three solutions of the Burgers equation with $|\bar e,\bar u,\bar y|\ne0$.
This results in \emph{item~4} of Theorem~\ref{Theorem_RO_subsets}.
The associated nonclassical reductions were discussed in~\cite[Section~3.3]{poch2013a}.

\subsection{\texorpdfstring{Case $\boldsymbol{\xi^1=0}$ with $\boldsymbol{\xi^0_{xx}\ne0}$}{Case xi1=0 with xi0xx not=0}}

The condition $\xi^1=0$  implies that $\xi=\xi^0(t,x)$ and $\eta=\eta^1(t,x)u+\eta^0(t,x)$.
Substituting these expression into the system~\eqref{de7}--\eqref{GBEROde}
and splitting~\eqref{de7} with respect to~$u$,
we obtain $\eta^1_x=0$, i.e.\ $\eta^1=\eta^1(t)$, and
\begin{gather}
\eta^0_x+\eta^1\xi^0_x=(\eta^1)^2-\eta^1_t,\label{case0neqdeq1}
\\
f_t+f_x\xi^0-f(\eta^1+\xi^0_x)=0,\label{case0neqdeq2}
\\
\xi^0_t+\xi^0\xi^0_x+f\xi^0_{xx}=\eta^0+\eta^1\xi^0,\label{case0neqdeq3}
\\
\eta^0_t+\eta^0\xi^0_x+f\eta^0_{xx}=\eta^1\eta^0.\label{case0neqdeq4}
\end{gather}

\begin{lemma}
For the case ${\xi^1=0}$, ${\xi^0_{xx}\ne0}$
the component~$\eta$ vanishes modulo~$G^\sim$.
\end{lemma}
\begin{proof}
For convenience we denote $\alpha:=(\eta^1)^2-\eta^1_t$.
The equation~\eqref{case0neqdeq1} can be easily integrated with respect to~$x$,
\begin{gather}
\label{eq_on_alpha_beta}
\eta^0+\eta^1\xi^0=\alpha x+\beta,
\end{gather}
where an arbitrary function $\beta=\beta(t)$ arises in the course of the integration.
After eliminating~$\eta^0$ in view of~\eqref{eq_on_alpha_beta} and recombining,
the equations~\eqref{case0neqdeq3} and~\eqref{case0neqdeq4} take the form
\begin{gather}
\label{eq_on_xi0}
\xi^0_t+\xi^0\xi^0_x+f\xi^0_{xx}=\alpha x+\beta,
\\
\label{eq_on_alpha}
\left((\alpha x+\beta)\xi^0\right)_x=2\eta^1(\alpha x+\beta)-(\alpha_tx+\beta_t).
\end{gather}

In fact, $\alpha=0$ and $\beta=0$
for any solution of the system~\eqref{case0neqdeq1}--\eqref{case0neqdeq4}.
Indeed, suppose that $\alpha\ne0$.
Integrating~\eqref{eq_on_alpha} with respect to~$x$, we get
\begin{gather*}
\xi^0=\dfrac{\left(\alpha\eta^1-\frac12\alpha_t\right)x^2+(2\eta^1\beta-\beta_t)x+\delta}
{\alpha x+\beta}=\gamma^1x+\gamma^0+\frac{\mu}{\alpha x+\beta},
\end{gather*}
where $\delta$ is an arbitrary function of~$t$ that arises in the course of integration, and
\begin{gather*}
\gamma^1=\eta^1-\frac{\alpha_t}{2\alpha},
\qquad
\gamma^0=\frac{\eta^1\beta-\beta_t}{\alpha}+\frac{\alpha_t\beta}{2\alpha^2},
\qquad
\mu=\delta+\frac{\beta}{\alpha}(\beta_t-\eta^1\beta)-\frac{\alpha_t\beta^2}{2\alpha^2}.
\end{gather*}
Then the equation~\eqref{eq_on_xi0} gives an expression for~$f$,
which is a~polynomial in~$(\alpha x+\beta)$,
and the coefficient of~$(\alpha x+\beta)^0$ is $\mu/(2\alpha)$.
Substituting the expressions for~$\xi^0$ and~$f$ into~\eqref{case0neqdeq2},
we derive the condition that a~polynomial in $(\alpha x+\beta)$ and $(\alpha x+\beta)^{-1}$
with coefficients depending on~$t$ identically equals zero.
This means that all the coefficients of this polynomial vanish.
In~particular, the coefficient of the lowest power~$(\alpha x+\beta)^{-2}$ is $\mu^2/2$.
Hence $\mu=0$ and then
we have $\xi^0=\gamma^1x+\gamma^0$, which contradicts the assumption $\xi^0_{xx}\ne0$.

Knowing that $\alpha=0$, we differentiate the equation~\eqref{eq_on_alpha} with respect to~$x$
and obtain $\beta\xi^0_{xx}=0$.
Therefore, $\beta=0$ in view of the assumption $\xi^0_{xx}\ne0$.

In view of the notation of $\alpha$ and the equation~\eqref{eq_on_alpha_beta},
the condition $\alpha=\beta=0$ implies that $\eta^1_t=(\eta^1)^2$ and $\eta^0=-\eta^1\xi^0$.
The solutions of the equation $\eta^1_t=(\eta^1)^2$ are $\eta^1=0$ and $\eta^1=-(t+c)^{-1}$,
where $c$ is an integration constant.
The second value of $\eta^1$ can be set to zero using the equivalence transformation
$\tilde t=-(t+c)^{-1}$,
$\tilde x=x(t+c)^{-1}$,
$\tilde u=(t+c)u-x$,
$\tilde f=f$.
\end{proof}

As $\eta=0$, the system of determining equations~\eqref{case0neqdeq1}--\eqref{case0neqdeq4}
reduces to the system
\begin{gather}
f_t+\xi f_x-\xi_xf=0,\label{PFDEcase_deq1}
\\
\xi_t+\xi\xi_x+f\xi_{xx}=0,\label{PFDEcase_deq2}
\end{gather}
which is well determined as it consists of two differential equations in two unknown functions, $f=f(t,x)$ and $\xi=\xi^0(t,x)$,
and definitely has no nontrivial differential consequences.

We write the equation~\eqref{PFDEcase_deq1} in conserved form, $(1/f)_t+(\xi/f)_x=0$,
and use it to introduce the potential $\theta=\theta(t,x)$
which is defined by the system $\theta_x=-1/f$, $\theta_t=\xi/f$, 
and thus $f=-1/\theta_x$, $\xi=-\theta_t/\theta_x$.
Substituting these expressions for~$f$ and~$\xi$ into~\eqref{PFDEcase_deq2} and integrating,
we obtain
\begin{gather}\label{GFDE}
\theta_t=\frac{\theta_{xx}}{\theta_x}+h(\theta)\theta_x,
\end{gather}
where~$h$ is an arbitrary smooth function of~$\theta$~\cite{poch2012a}.
As a~result, the system~\eqref{PFDEcase_deq1}--\eqref{PFDEcase_deq2}
reduces to the equation~\eqref{GFDE}.
For any~$h$ and for an arbitrary solution $\theta$ of the equation~\eqref{GFDE}, the vector field
\begin{gather*}
Q^\theta=\p_t-\frac{\theta_t}{\theta_x}\p_x
\end{gather*}
is a~reduction operator of the equation~$\mathcal L_{f^\theta}$ with $f^\theta=-1/\theta_x$,
which proves \emph{item~3} of Theorem~\ref{Theorem_RO_subsets}.

The impossibility of complete explicit description of reduction operators in this case,
which is equivalent to the problem of finding the general solution of the system~\eqref{PFDEcase_deq1}--\eqref{PFDEcase_deq2},
was signed out in~\cite{poch2012a}.
Nevertheless, new solutions of equations from the class~\eqref{GBE}
were constructed ibid by nonclassical reduction
via establishing the connection between
the system~\eqref{PFDEcase_deq1}--\eqref{PFDEcase_deq2}
and the potential fast diffusion equation~\eqref{GFDE} with $h=0$
and using a~number of already known exact solutions of the latter equation~\cite{popo2007a}.

\subsection{\texorpdfstring{Case $\boldsymbol{\xi^1=0}$ with $\boldsymbol{\xi^0_{xx}=0}$}{Case xi1=0 with xi0xx=0}}

We prove that in this case reduction operators are equivalent to Lie symmetry operators.
This assertion was first stated in~\cite[Section~2]{poch2012a}.
Under the condition $\xi^0_{xx}=0$, the equation~\eqref{case0neqdeq1} implies $\eta^0_{xx}=0$.
Hence $\xi(t,x)=\xi^0(t,x)=\xi^{01}(t)x+\xi^{00}(t)$ and $\eta(t,x)=\eta^1(t)u+\eta^{01}(t)x+\eta^{00}(t)$.
In terms of the new parameter-functions, the determining equations~\eqref{case0neqdeq1}--\eqref{case0neqdeq4} take the form
\begin{gather}
\xi^{01}_t=(\eta^1-\xi^{01})\xi^{01}+\eta^{01},\label{firsteq}
\\
\xi^{00}_t=(\eta^1-\xi^{01})\xi^{00}+\eta^{00},\label{xiunusedeq}
\\
\eta^{1\phantom0}_t=(\eta^1-\xi^{01})\eta^1-\eta^{01},\label{thirdeq}
\\
\eta^{01}_t=(\eta^1-\xi^{01})\eta^{01},\label{fourtheq}
\\
\eta^{00}_t=(\eta^1-\xi^{01})\eta^{00},\label{etaunusedeq}
\\
\eta^1=\frac{f_t}{f}+\frac{f_x}{f}\left(\xi^{01}x+\xi^{00}\right)-\xi^{01}.\label{classifeq}
\end{gather}
The next step is to modify the determining equations by the substitution $\xi^{01}=\varphi_t/\varphi$ and $\eta^1=-\psi_t/\psi$,
where $\varphi$ and $\psi$ are smooth functions of~$t$ with $\varphi\psi\ne0$.
Specifically, from~\eqref{fourtheq} we get the equation
\begin{gather*}
\dfrac{\eta^{01}_t}{\eta^{01}}+\frac{\varphi_t}{\varphi}+\frac{\psi_t}{\psi}=0,
\end{gather*}
which integrates to the condition $\eta^{01}\varphi\psi=a_0=\const$ and thus implies
\begin{gather}
\label{GBE_case00_26}
\eta^{01}=\frac{a_0}{\varphi\psi}.
\end{gather}
Substituting the expression~\eqref{GBE_case00_26} for $\eta^{01}$
into the modified equations~\eqref{firsteq} and~\eqref{thirdeq},
we derive $\varphi_{tt}\psi+\varphi_t\psi_t=a_0$ and $\varphi\psi_{tt}+\varphi_t\psi_t=a_0$,
or, after integration,
\begin{gather}
\label{GBE_case00_2728}
\varphi_t\psi=a_0t+a_1,
\qquad
\varphi\psi_t=a_0t+a_2,
\end{gather}
respectively.
The sum of the equations~\eqref{GBE_case00_2728} is directly integrated to
\begin{gather}
\label{GBE_case00_29}
\varphi\psi=a_0t^2+(a_1+a_2)t+a_3,
\end{gather}
where $a_1$, $a_2$ and $a_3$ are arbitrary constants
with $(a_0,a_1+a_2,a_3)\ne(0,0,0)$ since $\varphi\psi\ne0$.
Then the equation~\eqref{GBE_case00_26} leads to
\begin{gather*}
\eta^{01}=\dfrac{a_0}{a_0t^2+(a_1+a_2)t+a_3}.
\end{gather*}
Dividing each of the equations~\eqref{GBE_case00_2728} by~\eqref{GBE_case00_29} we obtain
\begin{gather*}
\xi^{01}=\dfrac{\varphi_t}{\varphi}=\frac{a_0t+a_1}{a_0t^2+(a_1+a_2)t+a_3},
\qquad
\eta^1=-\dfrac{\psi_t}{\psi}=-\frac{a_0t+a_2}{a_0t^2+(a_1+a_2)t+a_3}.
\end{gather*}
We sequentially integrate the two still unused equations~\eqref{etaunusedeq} and~\eqref{xiunusedeq} to derive
\begin{gather*}
\eta^{00}=\dfrac{a_4}{a_0t^2+(a_1+a_2)t+a_3},
\qquad
\xi^{00}=\dfrac{a_4t+a_5}{a_0t^2+(a_1+a_2)t+a_3}
\end{gather*}
with two more integration constants $a_4$ and $a_5$.

Finally, in view of the obtained representations of $\xi$ and $\eta$ we have the family of vector fields
\begin{gather*}
Q^a=\partial_t+\frac{(a_0t+a_1)x+a_4t+a_5}{a_0t^2+(a_1+a_2)t+a_3}\partial_x
+\frac{-(a_0t+a_2)u+a_0x+a_4}{a_0t^2+(a_1+a_2)t+a_3}\partial_u,
\end{gather*}
where the arbitrary constant tuple $a=(a_0,\ldots,a_5)$ is defined up to a~nonzero multiplier,
and $(a_0,a_1+a_2,a_3)\ne(0,0,0)$.
Each $Q^a$ is equivalent (up to multiplication by nonvanishing function) to the vector field
\begin{equation*}
\tilde Q^a=\big(a_0t^2+(a_1{+}a_2)t+a_3\big)\partial_t+\big((a_0t{+}a_1)x+a_4t+a_5\big)\partial_x+\big({-}(a_0t{+}a_2)u+a_0x+a_4\big)\partial_u.
\end{equation*}
Note that $\tilde Q^a$ is of the general form obtained for Lie symmetry operators in Section~\ref{CSdeteqSECTION}.
For $Q^a$ to be a~reduction operator of~$\mathcal L_f$ with certain~$f$,
its coefficients should additionally satisfy the equation~\eqref{classifeq}.
This is equivalent to the fact that
the components of $\tilde Q^a$ satisfy the classifying condition~\eqref{CSCC} with the same~$f$.
In other words, the vector field $Q^a$ is a~reduction operator of
an equation~$\mathcal L_f$ if and only if
the equivalent vector field $\tilde Q^a$ is the Lie symmetry generator of the same equation,
which results in \emph{item~2} of Theorem~\ref{Theorem_RO_subsets}.

Lie reductions and Lie solutions of equations from the class~\eqref{GBE}
were considered in Section~\ref{SOLUTIONSsection}.

\section{Conservation laws and potential admissible transformations}\label{CLPSsection}

Given a~generalized Burgers equation~$\mathcal L_f$ of the form~\eqref{GBE},
in order to study its (local) conservation laws
we use the common technique based on the notion of characteristic~\cite{olve1993b}
which was developed by Vinogradov~\cite{vino1984a}.
As this equation is a~second-order quasilinear evolution equation,
it is sufficient to consider the characteristics of conservation laws that depend only on~$t$, $x$ and~$u$,
$\lambda=\lambda(t,x,u)$; see~\cite[Corollary 2]{popo2008d}.

The necessary condition for a~conservation-law characteristic $\lambda$ of~$\mathcal L_f$ to be its co-symmetry leads to the determining equation
$-{\rm D}_t\lambda+{\rm D}_x^2(f\lambda)+u_x\lambda-{\rm D}_x(u\lambda)=0$ holding on solutions of~$\mathcal L_f$.
After the substitution $u_t=-uu_x-fu_{xx}$ we split this determining equation with respect to $u_{xx}$, which gives $\lambda_u=0$
and then simplifies it to $-\lambda_t+(f\lambda)_{xx}-u\lambda_x=0$.
Further splitting with respect to~$u$ results in $\lambda_x=0$ and $f_{xx}\lambda=\lambda_t$.
Hence $f_{xxx}\lambda=0$.
This means that the equation~$\mathcal L_f$ possesses nontrivial conservation laws if and only if $f_{xxx}=0$,
i.e.\
\begin{gather}\label{fquadratic}
f=f^2(t)x^2+f^1(t)x+f^0(t),
\end{gather}
where the coefficients $f^2$, $f^1$ and $f^0$ are smooth functions of~$t$ not vanishing simultaneously.
Then
\begin{gather*}
\lambda=\lambda(t)=e^{\int f_{xx}\d t}
\end{gather*}
is a~unique linearly independent characteristic of the equation~$\mathcal L_f$. 
Hereinafter, an integral with respect to~$t$ means a~fixed antiderivative. 
In other words, the space of conservation laws of~$\mathcal L_f$ is one-dimensional.
After multiplying by $\lambda$ the equation~$\mathcal L_f$ can be rewritten in the conserved~form
\begin{gather*}
{\rm D}_t\left(\lambda u\right)+\lambda {\rm D}_x\left(\frac12u^2+fu_x-f_xu\right)=0.
\end{gather*}
Using this form of~$\mathcal L_f$, we introduce the potential $v=v(t,x)$ defined by the system
\begin{gather}\label{GBE_PotentialSystem}
v_x=\lambda u,
\qquad
v_t=-\lambda\left(\frac12u^2+fu_x-f_xu\right).
\end{gather}
which is called a~potential system for~$\mathcal L_f$ and is denoted by~$\mathcal R_{f,\lambda}$.
After excluding the dependent variable~$u$ from~$\mathcal R_{f,\lambda}$,
we obtain the associated potential equation~$\mathcal P_{f,\lambda}$ for~$\mathcal L_f$,
\begin{gather}\label{GBE_PotentialEquation}
v_t+\frac1{2\lambda}v_x^{\,\,2}+fv_{xx}-f_xv_x=0
\quad
\text{with}
\quad
f_{xxx}=0, \quad \lambda_t=f_{xx}\lambda,\quad f\lambda\ne0.
\end{gather}

Since it is impossible to choose a canonical representative 
in a one-dimensional space of conservation-law characteristics 
of the equation~$\mathcal L_f$ with general~$f$ quadratic in~$x$, 
in the course of the consideration 
of the class~\eqref{GBE_PotentialSystem} of potential systems,~$\mathcal R_{f,\lambda}$, 
and of the class~\eqref{GBE_PotentialEquation} of potential equations,~$\mathcal P_{f,\lambda}$, 
it is necessary to extend the arbitrary element with~$\lambda$.
The set~$\breve{\mathcal S}$ run by the extended arbitrary element~$(f,\lambda)$ 
is defined by the auxiliary system $f_{xxx}=0$, $\lambda_t=f_{xx}\lambda$, $f\lambda\ne0$. 
Potential systems~$\mathcal R_{f,\lambda}$ and~$\mathcal R_{f,\smash{\tilde\lambda}}$ 
(resp.\ potential equations~$\mathcal P_{f,\lambda}$ and~$\mathcal P_{f,\smash{\tilde\lambda}}$) 
with linearly dependent~$\lambda$ and~$\tilde\lambda$ are assumed to be gauge-equivalent 
in the sense of potential systems (resp.\ potential equations). 
Therefore, we associate the single equation~$\mathcal L_f$ 
with the set of gauge-equivalent potential systems $\{\mathcal R_{f,\lambda}\}$ and
with the set of gauge-equivalent potential equations~$\{\mathcal P_{f,\lambda}\}$.

We can also study potential conservation laws of the equation~$\mathcal L_f$, which are
local conservation laws of the corresponding potential system~$\mathcal R_{f,\lambda}$,
or, equivalently, the potential equation~$\mathcal P_{f,\lambda}$.%
\footnote{%
There exists an isomorphism between
the spaces of conservation laws of the equation~\eqref{GBE_PotentialEquation}
and of the system~\eqref{GBE_PotentialSystem}.
}
Let $\mu$ be the reduced characteristic of a~local conservation law for the potential equation~$\mathcal P_{f,\lambda}$.
Then $\mu=\mu(t,x,v)$ \cite[Corollary 2]{popo2008d}.
We write the equation
$-{\rm D}_t\mu+{\rm D}_x^2(f\mu)-{\rm D}_x((\frac{1}{\lambda}v_x-f_x)\mu)=0$
holding on solutions of~$\mathcal P_{f,\lambda}$
and solve it with respect to $\mu$,
which gives only zero solutions if $f\ne\const$.
As a~result, we conclude that
the equation~$\mathcal L_f$ with $f\ne\const$ has no nonzero potential conservation laws,
and thus $\mathcal R_{f,\lambda}$ is the only canonical potential system for~$\mathcal L_f$.
The classical Burgers equation (for which $f=\const$) possesses potential conservation laws,
which are local conservation laws of the corresponding potential equation~$\mathcal P_{f,\lambda}$
with reduced characteristics of the form $\mu=\psi(t,x)e^{\frac{v}{2f}}$,
where $\psi(t,x)$ is an arbitrary solution of the linear heat equation $\psi_t=f\psi_{xx}$.
Since the potential Burgers equation~$\mathcal P_{f,\lambda}$
is similar to the linear heat equation,
the classical Burgers equation admits no higher-level potential conservation laws
\mbox{\cite[Theorem~5]{popo2008a}}.

Therefore, the class~\eqref{GBE}$_{f_{xxx}=0}$ of equations of the form~\eqref{GBE} with~$f$ quadratic in~$x$
is naturally partitioned into the subclasses 
\[\mathcal L^1=\{\mathcal L_f \mid f=\const\} \quad\mbox{and}\quad \mathcal L^2=\{\mathcal L_f \mid f_{xxx}=0,\, f\ne\const\}.\]
Both these subclasses are closed under the action of the equivalence group $G^\sim$
and, therefore, are normalized as well as the class~\eqref{GBE}
in view of~\cite[Proposition~4]{popo2010a}.

Admissible transformations and the equivalence groupoid
of the class of potential systems~\eqref{GBE_PotentialSystem}
can be called \emph{potential admissible transformations}
and the \emph{potential equivalence groupoid} of the class~\eqref{GBE}$_{f_{xxx}=0}$,
respectively.

\begin{lemma}\label{lem:PotAdmTransOfGBEs}
The equivalence groupoids of the class of potential systems~\eqref{GBE_PotentialSystem}
and of the class of potential equations~\eqref{GBE_PotentialEquation} are isomorphic.
\end{lemma}

\begin{proof}
We fix any two values of the arbitrary element from the set~$\breve{\mathcal S}$, 
$(f,\lambda)$ and $(\tilde f,\tilde\lambda)$.
Every point transformation between the equations~$\mathcal P_{f,\lambda}$ and $\mathcal P_{\tilde f,\tilde\lambda}$
is prolonged to~$u$ according to the equation $u=v_x/\lambda$,
which gives a~point transformation between the systems~$\mathcal R_{f,\lambda}$ and $\mathcal R_{\tilde f,\tilde\lambda}$.
Conversely, let $\varphi$ be a~point transformation
between the systems~$\mathcal R_{f,\lambda}$ and $\mathcal R_{\tilde f,\tilde\lambda}$,
\begin{gather*}
\varphi\colon\quad
\tilde t=T(t,x,u,v),
\quad
\tilde x=X(t,x,u,v),
\quad
\tilde u=U(t,x,u,v),
\quad
\tilde v=V(t,x,u,v).
\end{gather*}
Substituting $u=v_x/\lambda$ into~$T$, $X$ and~$V$ and neglecting the transformation component of~$\varphi$ for~$u$,
we obtain a~contact transformation between~$\mathcal P_{f,\lambda}$ and $\mathcal P_{\tilde f,\tilde\lambda}$,
which is necessarily induced by a~point transformation
since both the equations~$\mathcal P_{f,\lambda}$ and $\mathcal P_{\tilde f,\tilde\lambda}$ are second-order evolution equations
that are linear with respect to the second derivative~$u_{xx}$~\cite[Proposition 2]{popo2008d}.
\end{proof}

In view of Lemma~\ref{lem:PotAdmTransOfGBEs}, we can deal
with the class of potential equations~\eqref{GBE_PotentialEquation}
instead of the class of potential systems~\eqref{GBE_PotentialSystem}
when studying potential admissible transformations of the class~\eqref{GBE}$_{f_{xxx}=0}$.

\begin{corollary}\label{cor:PotSymsOfGBEs}
For each fixed value~$f$ with $f_{xxx}=0$,
there exists an isomorphism between the Lie invariance algebras~$\mathfrak g_{\mathcal P_{f,\lambda}}$ and~$\mathfrak g_{\mathcal R_{f,\lambda}}$
of the equation~$\mathcal P_{f,\lambda}$ and the system~$\mathcal R_{f,\lambda}$.
\end{corollary}

This isomorphism is provided by the projection to the space of $(t,x,v)$
when mapping~$\mathfrak g_{\mathcal R_{f,\lambda}}$ onto~$\mathfrak g_{\mathcal P_{f,\lambda}}$
and by the prolongation according to the equation $u=v_x/\lambda$ for the inverse mapping.

In order to find the equivalence groupoid of the class~\eqref{GBE_PotentialEquation},
it is convenient to reparameterize this class
by assuming the coefficients $f^2$, $f^1$ and $f^0$ of~$f$ in~\eqref{fquadratic}
to be new arbitrary elements.
We use both the parameterizations simultaneously.
 
The above partition of the class~\eqref{GBE}$_{f_{xxx}=0}$
hints at the partition of the class~\eqref{GBE_PotentialEquation}
into the subclasses 
\[\mathcal P^1=\{\mathcal P_{f,\lambda} \mid f=\const\} \quad\mbox{and}\quad \mathcal P^2=\{\mathcal P_{f,\lambda} \mid f_{xxx}=0,\, f\ne\const\},\]
which have essentially different transformational properties
and are potential counterparts of the subclasses~$\mathcal L^1$ and~$\mathcal L^2$,
respectively.

\begin{proposition}\label{PotClassGroupoid}
The equivalence groupoid of the subclass~$\mathcal P^2$ 
consists of the triples of the form $(f,\varphi,\tilde f)$ 
where the components of the transformation~$\varphi$ are 
\begin{gather*}
\tilde t=\frac{\alpha t+\beta}{\gamma t+\delta},
\quad
\tilde x=\frac{\kappa x+\mu_1t+\mu_0}{\gamma t+\delta},
\\
\tilde v=c_0\left(v
-\frac{\gamma\lambda(t)}{2(\gamma t+\delta)}x^2
+\frac{\delta\mu_1-\gamma\mu_0}{\kappa(\gamma t+\delta)}\lambda(t)x
+V^0(t)\right)\!,
\end{gather*}
the relation between the source and target arbitrary elements is given by 
\begin{gather}
\begin{split}\label{PotClassf}
&\tilde f^2=\frac{(\gamma t+\delta)^2}{\Delta}f^2,
\quad
\tilde f^1=\frac{\gamma t+\delta}{\Delta}\big(\kappa f^1-2(\mu_1 t+\mu_0)f^2\big),
\\
&\tilde f^0=\frac{1}{\Delta}\big(\kappa^2f^0-\kappa(\mu_1 t+\mu_0)f^1+(\mu_1 t+\mu_0)^2f^2\big),
\quad
\text{i.e.}
\quad
\tilde f=\frac{\kappa^2}{\Delta}f,
\quad
\tilde\lambda=c_0\frac{\Delta}{\kappa^2}\lambda,
\end{split}
\end{gather}
$\alpha$, $\beta$, $\gamma$, $\delta$, $\mu_0$, $\mu_1$ and $\kappa$ are arbitrary constants
defined up to a~nonzero multiplier with $\Delta:=\alpha\delta-\beta\gamma\ne0$ and $\kappa\ne0$;
$c_0$ and $\sigma$ are arbitrary constants with $c_0\ne0$, and
\begin{gather}\label{PotClassV0}
V^0=
\!\int\!\left(\frac{(\delta\mu_1-\gamma\mu_0)^2}{2\kappa^2(\gamma t+\delta)^2}
+\frac{\delta\mu_1-\gamma\mu_0}{\kappa(\gamma t+\delta)}f^1(t)+\frac\gamma{\gamma t+\delta}f^0(t)\right)\lambda(t)\d t+\sigma,
\end{gather}
\end{proposition}

\begin{proof}
Each equation~$\mathcal P_{f,\lambda}$ from the class~\eqref{GBE_PotentialEquation}
is mapped by the point transformation
\begin{gather}\label{Tr_PotClasses}
\psi_{f,\lambda}\colon
\qquad
\hat t=\frac12\int\lambda \d t,
\qquad
\hat x=\lambda x+\int f^1\lambda \d t,
\qquad
\hat v=v,
\end{gather}
parameterized by~$f$ (more precisely, by $f^2$ and $f^1$) and $\lambda$
to a~shorter equation~$\hat{\mathcal P}_{\hat f}$ from the class
\begin{gather}\label{PotClass}
\hat v_{\hat t}+\hat v_{\hat x}^{\,\,2}+\hat f\hat v_{\hat x\hat x}=0
\qquad
\text{with}
\qquad
\hat f\ne0,\quad
\hat f_{\hat x\hat x\hat x}=0.
\end{gather}
The image of~$f$ under the mapping is computed by
$\hat f(\hat t,\hat x)=2\lambda(t)f(t,x)$.
Equivalently, we can compute the images of the coefficients of~$f$.
The subclasses~$\mathcal P^1$ and~$\mathcal P^2$ of the class~\eqref{GBE_PotentialEquation}
are then mapped to the subclasses
$\hat{\mathcal P}^1=\{\hat{\mathcal P}_{\hat f} \mid \hat f=\const\}$
and~$\hat{\mathcal P}^2=\{\hat{\mathcal P}_{\hat f} \mid \hat f_{\hat x\hat x\hat x}=0,\,\hat f\ne\const\}$
of the class~\eqref{PotClass}, respectively.
The transformational properties of the subclass~\smash{$\hat{\mathcal P}^2$} essentially differ from those
of the subclass~\smash{$\hat{\mathcal P}^1$},
and there are no point transformations between equations from the different subclasses;
see \cite{poch2015a} (note that the subclasses notation therein is slightly different).
The transformational part of each admissible transformation
of the subclass~$\hat{\mathcal P}^2$ has the following properties:
\begin{itemize}\setSkips
\item the transformational component for~$\hat t$ depends only on~$\hat t$,
\item the transformational component for~$\hat x$ depends only on $(\hat t,\hat x)$ and is linear in~$\hat x$, 
\item the transformational component for~$\hat v$ is linear in~$\hat v$ with a~constant coefficient of~$\hat v$
and is quadratic in~$\hat x$.
\end{itemize}

Consider an arbitrary admissible transformation
$(f,\varphi,\tilde f)$
of the subclass~$\mathcal P^2$.
Denote
\begin{gather*}
\hat f:=\psi_{f,\lambda}f,
\qquad
\hat{\tilde f}:=\psi_{f,\lambda}\tilde f,
\qquad
\hat\varphi:=\psi_{\tilde f,\tilde\lambda}\varphi\psi_{f,\lambda}^{-1},
\qquad
\text{and hence}
\qquad
\varphi=\psi_{\tilde f,\tilde\lambda}^{-1}\hat\varphi\psi_{f,\lambda}.
\end{gather*}
Then the triple $(\hat f,\hat\varphi,\hat{\tilde f})$ is an admissible transformation
of the subclass~$\hat{\mathcal P}^2$,
and hence its transformational part $\hat\varphi$ has the properties listed above.
In view of~\eqref{Tr_PotClasses} and the relation of $\varphi$ with $\hat\varphi$,
the point transformation $\varphi$ has the same properties, i.e.\
\begin{gather*}
\varphi\colon
\quad
\tilde t=T(t),
\quad
\tilde x=X^1(t)x+X^0(t),
\quad
\tilde v=c_0\big(v+V^2(t)x^2+V^1(t)x+V^0(t)\big),
\end{gather*}
where $T_tX^1c_0\ne0$.
We substitute the expressions of tilded entities in terms of untilded ones, including the derivatives
\begin{gather*}
\tilde v_{\tilde t}=\frac{c_0}{T_t}\left(v_t+V^2_tx^2+V^1_tx+V^0_t-\frac{X^1_tx+X^0_t}{X^1}(v_x+2V^2x+V^1)\right),
\\
\tilde v_{\tilde x}=c_0\frac{v_x+2V^2x+V^1}{X^1},
\quad
\tilde v_{\tilde x \tilde x}=c_0\frac{v_{xx}+2V^2}{(X^1)^2},
\end{gather*}
and also $v_t=-\frac1{2\lambda}v_x^2-fv_{xx}+f_xv_x$
into~$\mathcal P_{\tilde f}$, and split the obtained equation with respect to $v_{xx}$, $v_x$ and~$x$.
After simplifying, this gives
\begin{gather}
\label{PotClassFourDeq}
\tilde f=\frac{(X^1)^2}{T_t}f,
\quad
\tilde\lambda=\frac{c_0T_t}{(X^1)^2}\lambda,
\quad
V^2=\frac{c_0}2\frac{X^1_t}{X^1}\lambda,
\quad
V^1=c_0\frac{X^0_t}{X^1}\lambda,
\\
\left(\frac{X^1_t}{(X^1)^2}\right)_t=0,
\quad
\left(\frac{X^0_t}{(X^1)^2}\right)_t=0,
\quad
V^0_t=c_0\lambda\left(\frac12\left(\frac{X^0_t}{X^1}\right)^2+\frac{X^0_t}{X^1}f^1-\frac{X^1_t}{X^1}f^0\right).
\label{PotClassThreeDeq}
\end{gather}
The first equation of~\eqref{PotClassFourDeq} implies 
the relation between coefficients of~$f$ and~$\tilde f$,
\begin{gather}\label{PotClassfinTX}
\tilde f^2=\frac1{T_t}f^2,
\quad
\tilde f^1=\frac{X^1}{T_t}f^1-2\frac{X^0}{T_t}f^2,
\quad
\tilde f^0=\frac{(X^1)^2}{T_t}f^0-\frac{X^0X^1}{T_t}f^1+\frac{(X^0)^2}{T_t}f^2.
\end{gather}
Since $\lambda=e^{2\int f^2\d t}$,
$\tilde\lambda=e^{2\int\tilde f^2\d\tilde t}$
and $T_t\tilde f^2=f^2$, we obtain
\begin{gather*}
\left(\frac{\tilde\lambda}{\lambda}\right)_t=2\left(T_t\tilde f^2-f^2\right)\frac{\tilde\lambda}{\lambda}=0,
\quad
\text{i.e.}
\quad
\frac{\tilde\lambda}{\lambda}=\const.
\end{gather*}
Hence the second equation of~\eqref{PotClassFourDeq} gives
$T_t/(X^1)^2=\tilde\lambda/(c_0\lambda)=\const$.
The first equation of~\eqref{PotClassThreeDeq} can be represented as
$(1/X^1)_{tt}=0$, which means that
$1/X^1$ is a~linear function of~$t$.
Then the relation $T_t=\const\cdot(X^1)^2$ implies that $T$ is fractional linear in~$t$.
As a~result, we derive
\begin{gather*}
T=\frac{\alpha t+\beta}{\gamma t+\delta},
\quad
X^1=\frac{\kappa}{\gamma t+\delta},
\quad
X^0=\frac{\mu_1 t+\mu_0}{\gamma t+\delta},
\end{gather*}
where the expression for $X^0$ is obtained from the second equation of~\eqref{PotClassThreeDeq}.
The last equation of~\eqref{PotClassThreeDeq} leads, after integration, to the expression~\eqref{PotClassV0}.
Finally, from~\eqref{PotClassfinTX} we have~\eqref{PotClassf}.
\end{proof}

\begin{corollary}\label{cor:NoPotAdmTransBetweenSubclassesOfGBEs}
There are no point transformations relating equations from the different subclasses~$\mathcal P^1$ and~$\mathcal P^2$.
\end{corollary}

Since admissible transformations of the subclass~$\mathcal P^2$ are uniformly parameterized 
by the extended arbitrary element in a nonlocal way, 
we may say that this subclass is normalized in the extended generalized sense 
although there is no point transformation group underlying this normalization in a natural way. 

Comparing the transformations in Theorem \ref{GBEgroupoid} and Proposition \ref{PotClassGroupoid}
and knowing the connection~\eqref{GBE_PotentialSystem} between~$u$ and~$v$,
we conclude that the equivalence groupoids
of the classes~$\mathcal L^2$ and~$\mathcal P^2$
are isomorphic up to gauge shifts $\tilde v=v+\sigma$ and gauge scalings $(\tilde v,\tilde\lambda)=(c_0v,c_0\lambda)$.
The same is true for the equivalence groups of the classes~$\mathcal L^2$ and~$\mathcal P^2$
although these equivalence groups are of different kinds.
This also implies that for each nonconstant~$f$ with $f_{xxx}\ne0$
the Lie symmetry groups of the equations~$\mathcal L_f$ and~$\mathcal P_{f,\lambda}$
are isomorphic up to shifts of~$v$,
which belong to the kernel group of~$\mathcal P^2$,
and so the group classifications for the classes~$\mathcal L^2$ and~$\mathcal P^2$ are equivalent.
Since the class~$\mathcal L^2$ is closed under the action of $G^\sim$,
its group classification can be easily separated out from the group classification
of the entire class~\eqref{GBE}, which is given in Table~\ref{TableSubalgebras}.

\begin{proposition}\label{Prop_PotSymsWithQuadratic_f}
Potential admissible transformations of the class~$\mathcal L^2$
are induced by its usual admissible transformations.
In particular, equations from this class have no nontrivial potential symmetries.
\end{proposition}

Consider the potential counterpart~$\mathcal P^1$ of the class~$\mathcal L^1$.
For constant values of~$f$, the transformation~\eqref{Tr_PotClasses}
degenerates to a~simple rescaling of~$t$.
Hence the class~$\mathcal P^1$ in fact coincides, up to the scaling $\hat t=\frac12t$,
with the class~$\hat{\mathcal P}^1$, see equation~\eqref{PotClass}.
The equivalence groupoid of the class~$\hat{\mathcal P}^1$ is described in \cite[Proposition 2]{poch2015a}.
We modify it
up to the multipliers $\frac12$ and $\kappa$
for better consistence with the present paper.

\begin{proposition}\label{Prop_PotClassConst_G}
The subclass~$\mathcal P^1=\{\mathcal P_{f,\lambda} \mid f=\const\}$ of the class~\eqref{GBE_PotentialEquation}
is normalized in the generalized sense.
Its generalized equivalence group is constituted by the transformations
\begin{gather*}
\tilde t=\frac{\alpha t+\beta}{\gamma t+\delta},
\quad
\tilde x=\frac{\kappa x+\mu_1t+\mu_0}{\gamma t+\delta},
\quad
\tilde v=\frac{2\kappa^2f}{\Delta}\ln\left|F^1\Big(e^{\tfrac{v}{2f}}+F^0\Big)\right|,
\quad
\tilde f=\frac{\kappa^2}{\Delta}f,
\end{gather*}
where $\alpha$, $\beta$, $\gamma$, $\delta$, $\kappa$, $\mu_1$, $\mu_0$
are arbitrary constants with $\Delta:=\alpha\delta-\beta\gamma\ne0$ and $\kappa\ne0$,
the tuple $(\alpha,\beta,\gamma,\delta,\kappa)$ is defined up to a~nonzero multiplier,
\begin{gather*}
F^1=\begin{cases}
k\sqrt{|\gamma t+\delta|}\exp\left(
-\dfrac{(\gamma\kappa x-\mu_1\delta+\mu_0\gamma)^2}
{4f\kappa^2\gamma(\gamma t+\delta)}
\right), & \gamma\ne0,
\\[4mm]
k\exp\dfrac{2\kappa\mu_1x+\mu_1^2t}{4\kappa^2f}, & \gamma=0,
\end{cases}
\end{gather*}
$k$ is a~nonzero constant, and $F^0$ is a~solution of the linear equation
$F^0_t+fF^0_{xx}=0$.
\end{proposition}

Proposition~\ref{Prop_PotClassConst_G} is quite obvious since
the class~$\mathcal P^1$ is the orbit of any its single equation
under the action of a~scaling group.
The complicated form of admissible transformations in the class~$\mathcal P^1$
is a~consequence of the fact that
any equation~$\mathcal P_{f,\lambda}$ with $f=\const$ can be linearized to the heat equation $w_t+fw_{xx}=0$ by the variable change $w=e^{v/(2f)}$,
and thus general symmetry transformations of~$\mathcal P_{f,\lambda}$ are of complicated form and involve an arbitrary solution of the heat equation.
The potential symmetries of the classical Burgers equation are commonly known,
see, e.g., \cite[Example 2.42]{olve1993b} and~\cite{popo2005a}.
Up to linear combining, the nontrivial potential symmetries of~$\mathcal L_f$
that are associated with local conservation laws
are exhausted by the images of linear superposition symmetries of the corresponding linear heat equation $w_t+fw_{xx}=0$
under push-forwarding with $v=2f\ln w$ and prolonging to~$u$ according to $u=v_x$.

Propositions~\ref{PotClassGroupoid} and~\ref{Prop_PotClassConst_G}
jointly with Corollary~\ref{cor:NoPotAdmTransBetweenSubclassesOfGBEs}
give the complete description of the equivalence groupoid
of the class~\eqref{GBE_PotentialEquation}.

\section{Conclusion}

Extended symmetry analysis of the class~\eqref{GBE} of generalized Burgers equations
was merely one of the purposes of this paper.
Although the class~\eqref{GBE} looks quite simple,
it has several interesting specific properties
that are related to the field of group analysis of differential equations.
This is why the class~\eqref{GBE} is a~convenient object for testing various methods
that were recently developed in this field.
Moreover, the study of the class~\eqref{GBE} in the present paper
can serve as a~source of ideas for improving and modifying known techniques.

\looseness=-1
Although the arbitrary element~$f$ of the class~\eqref{GBE} depends on two arguments
and the usual equivalence group~$G^\sim$ of this class is finite-dimensional,
the class~\eqref{GBE} is normalized in the usual sense,
and thus it is a~rare bird among classes of differential equations
considered in the literature on symmetries.
Similar subclasses of variable-coefficient Korteweg--de Vries equations
were singled out in~\cite{gaze1992a}
in the course of symmetry analysis of a~wider class of KdV-like equations,
$u_t+f(t,x)uu_x+g(t,x)u_{xxx}=0$ with $fg\ne0$. 
These subclasses are associated with the constraints $f=x$, $f=x^{-1}$ and $f=1$, 
and the last subclass is the most relevant to the class~\eqref{GBE}.
Therein, Lie symmetries of equations from these three subclasses
were studied using an early version of the algebraic method of group classification.
See also a~modern treatment of these results in~\cite{vane2016a}.
At the same time, the algebraic method of group classification was
(implicitly or explicitly) used mostly for normalized classes
whose equivalence algebras are infinite-dimensional;
see \cite{basa2001a,bihl2012b,kuru2016a,popo2010a,popo2008a} and references therein.
The attempt of extending results of~\cite{gaze1992a}
to variable-coefficient Burgers equations
of the form $u_t+f(t,x)uu_x+g(t,x)u_{xx}=0$ with $fg\ne0$ in~\cite{qu1995b}
was not completely successful, in particular,
due to weaknesses of the used sets of subalgebras
of the related (finite-dimensional) algebras of vector fields.
This is why it was instructive to carry out the group classification of the class~\eqref{GBE}
accurately using various notions and techniques within the framework of the algebraic method,
including the ascertainment of normalization of the class~\eqref{GBE}
and the selection of appropriate subalgebras of (the projection of) its equivalence algebra.

For optimizing the procedure of Lie reductions of equations from the class~\eqref{GBE},
we have applied two special techniques.
The technique of classifying of Lie reductions of partial differential equations from a~normalized class
with respect to its equivalence group is completely original.
This technique allowed us to study Lie reductions for the whole class~\eqref{GBE} at once
but not separately for each case of Lie symmetry extensions listed in Table~\ref{TableSubalgebras}.
Moreover, the classification of appropriate subalgebras of the projection of the equivalence algebra
can then be adapted for the classification of Lie reductions.
Each subalgebra leads to a~set of equivalent ansatzes,
and the second technique is aimed to optimize the form of reduced equations
via selecting specific ansatzes among equivalent ones.
In other words, the selected ansatzes are not of the simplest form
but they are complicated as much as it is necessary
for simplifying the further study of reduced equations.

There are two kinds of Lie reductions of generalized Burgers equations
from the class~\eqref{GBE}, singular and regular ones.
The algebras~$\mathfrak g^{1.0}$ and~$\mathfrak g^{1.1}$
lead to singular reductions since the order of associated reduced equations, which equals one,
is less than the order of original equations, which equals two.
These reduced equations can easily be integrated
although this gives only trivial solutions of original equations.
These solutions linearly depend on~$x$
and are thus common for all equations from the class~\eqref{GBE}.
Lie reductions with the algebras \mbox{$\mathfrak g^{1.2}$--$\mathfrak g^{1.8}_a$} are regular.
Due to the selection of ansatzes,
the associated reduced equations are included in the single class~\eqref{RedEqClass},
which allowed us to unify symmetry analysis of reduced equations
and to completely describe  hidden symmetries of equations from the class~\eqref{GBE}.

Theorem~\ref{Theorem_RO_subsets} together with the preceding discussion of singular reduction operators
presents one of a~few examples existing in the literature,
where reduction operators are completely described for a~nontrivial class of differential equations
as well as nonclassical reductions result in effectively finding new exact solutions.
The technique of classifying reduction operators
with respect to the equivalence group of the class under consideration
plays an important but not crucial role here.

Three ``no-go'' cases of different nature are singled out in the course of the study.
For every equation~$\mathcal L_f$ from the class~\eqref{GBE},
the problem of describing its singular reduction operators is ``no-go''
since the corresponding ansatzes reduce~$\mathcal L_f$
to first-order ordinary differential equations.
This is a~general property of partial differential equations with two independent variables
that possess first co-order singular sets of vector fields
parameterized by single arbitrary functions of both independent and dependent variables
\cite[Section~8]{kunz2008b}.
Another ``no-go'' case is related to regular reduction operators of the classical Burgers equation.
It was studied in~\cite{arri2002a,mans1999a}
and can be explained by the linearization of the Burgers equation to the heat equation with the Hopf--Cole transformation
and the existence of the similar ``no-go'' case for regular reduction operators of linear evolution equations
\cite{fush1992e,popo2008b}.
The last ``no-go'' case arises due to the study of reduction operators for a~class of differential equations,
and this phenomenon was not discussed in the literature before~\cite{poch2012a}.
Involving the arbitrary element~$f$ to the determining equations for the components of regular reduction operators
leads to the necessity of solving the well-defined system~\eqref{PFDEcase_deq1}--\eqref{PFDEcase_deq2}.
At the same time, the last case is of most interest in spite of its ``no-go'' essence.
The system~\eqref{PFDEcase_deq1}--\eqref{PFDEcase_deq2} reduces to the single equation~\eqref{GFDE}.
Many exact solutions of the equation~\eqref{GFDE} are known for the value $h=0$,
which immediately results in exact solutions of various equations from the class~\eqref{GBE}.

We studied both local and potential conservation laws of these equations.
Only equations with $f_{xxx}=0$ admit nonzero conservation laws.
Using the subclass~\eqref{GBE}$_{f_{xxx}=0}$ of such equations as illustrating example,
we introduced the notion of potential equivalence groupoid of a~class of differential equations
and then computed the potential equivalence groupoid of the above subclass.
All assertions on potential symmetries of equations from the class~\eqref{GBE}
are direct consequences of the comparison of the usual and potential equivalence groupoids
of the subclass~\eqref{GBE}$_{f_{xxx}=0}$.

\vspace*{-1ex}
\subsection*{Acknowledgements}
The authors thank Vyacheslav Boyko for useful discussions.
The research of ROP was supported by project P25064 of FWF.

\footnotesize\frenchspacing

\end{document}